\newcommand{\fullv}[1]{#1}
\newcommand{\shortv}[1]{}

\shortv{
\documentclass{sig-alternate-05-2015}}

\fullv{
  \documentclass[11pt]{article}
  \usepackage{fullpage}
\usepackage{chicago}
\usepackage{hyperref}
\usepackage{amsthm}
\headheight \baselineskip
\headsep 1.5\baselineskip
\topmargin -0.375in
\textheight 9in
\textwidth 6.5in
\columnsep .25in
}

\usepackage{tonsmod}

\usepackage{verbatim}
\usepackage{amssymb}
\usepackage{amsmath}
\usepackage{enumitem}
\usepackage[noend]{algpseudocode}
\usepackage{algorithm}
\usepackage{verbatim}
\usepackage{amssymb}
\usepackage{amsmath}
\usepackage[noend]{algpseudocode}
\usepackage{algorithm}
\usepackage{enumitem}

\usepackage{xspace}

\newcommand{\osc}{\mathit{cons}}
\newcommand{\se}{\mathit{se}}
\newcommand{\F}{\mathcal{F}}
\newcommand{\commentout}[1]{}

\usepackage{thmtools, thm-restate}
\declaretheorem{theorem}

\declaretheorem{lemma}

\algblockdefx[Phase]{Phase}{EndPhase}
   [2]{\textbf{Phase} #1: #2}
   [0]{\textbf{EndPhase}}

\newcommand{\vsigma}{\vec{\sigma}}
\newcommand{\R}{\mathcal{R}}
\newcommand{\I}{\mathcal{I}}

\newcommand{\true}{\mathit{true}}
\newcommand{\false}{\mathit{false}}

\begin{document}

\shortv{



\CopyrightYear{2016}
\setcopyright{acmlicensed}
\conferenceinfo{PODC'16}{July 25--28, 2016, Chicago, IL, USA}
\isbn{978-1-4503-3964-3/16/07}
\acmPrice{\$15.00}
\doi{http://dx.doi.org/10.1145/2933057.2933088}

}

\shortv{
\numberofauthors{2} 
\author{
\alignauthor
Joseph Y. Halpern\\
       \affaddr{Cornell University}\\
       \email{halpern@cs.cornell.edu}
\alignauthor
Xavier Vila\c{c}a\\
       \affaddr{INESC-ID}
       \affaddr{Instituto Superior T\'{e}cnico}
       \affaddr{Universidade de Lisboa}\\
       \email{xvilaca@gsd.inesc-id.pt}
}
}

\title{Rational Consensus\shortv{: Extended Abstract}}

\fullv{
\author{Joseph Y. Halpern%
  \thanks{Supported in part by NSF grants 
IIS-0911036 and CCF-1214844, and by AFOSR grant
FA9550-12-1-0040,
and ARO grant W911NF-14-1-0017.}\\
Cornell University\\
halpern@cs.cornell.edu 
\and Xavier Vila\c{c}a%
\thanks{Supported in part by Funda\c c\~{a}o para a Ci\^{e}ncia e Tecnologia (FCT) 
through projects with references PTDC/EEI-SCR/1741/2014 (Abyss), UID/CEC/50021/2013,
and ERC-2012-StG-307732, and through the PhD grant SFRH/BD/79822/2011.}
\\
INESC-ID Lisboa\\
xvilaca@gsd.inesc-id.pt}
}

\maketitle

\fullv{
\thispagestyle{empty}}

\begin{abstract}
We provide a game-theoretic analysis of consensus,
assuming that processes are controlled by rational agents and may fail 
by crashing.  We consider agents that \emph{care only about consensus}:
that is, (a) an agent's utility 
depends only on the consensus value achieved (and not, for example, on
the number of messages the agent sends) and  (b) agents strictly prefer
reaching consensus to not reaching consensus.
We show that, under these assumptions, there is no 
\emph{ex post Nash Equilibrium}, even with only one failure. 
Roughly speaking, this means that there must always exist a
\emph{failure pattern} (a description of who fails, when they fail, and
which agents they do not send messages to in the round that they fail)
and initial preferences  for which an agent can gain by deviating.
On the other hand, if we assume that there is a distribution $\pi$ 
on the failure patterns and initial preferences, then under minimal
assumptions on $\pi$, there is a Nash equilibrium that tolerates $f$
failures (i.e., $\pi$ puts probability 1 on there being at most $f$
failures) if $f+1 < n$ (where $n$ is the total number of agents).
Moreover, we show that a slight extension of the Nash equilibrium
strategy is also a \emph{sequential} equilibrium (under the same
assumptions about the distribution $\pi$).
\end{abstract}

\shortv{
\commentout{



\terms{Theory}

}}

\shortv{}

\fullv{
\vspace{1in}


\newpage
\pagestyle{plain}
\setcounter{page}{1}
}

\section{Introduction}
Consensus is a fundamental problem in distributed computing;
it plays a key role in state machine replication,
transaction commitment, 
and many other tasks where agreement among processes is required. 
It is well known that consensus cannot be 
deterministically
achieved in asynchronous
systems \cite{FLP}, but can be achieved in synchronous systems even if
we allow Byzantine failures (see, e.g., \cite{Lyn97}).  The assumption
in all these solutions is that the reason that processes do not follow
the protocol is that they have been taken over by some adversary.
\shortv{\vfill\eject}

There has been a great deal of interest recently in viewing at least
some of the processes as being under the control of rational agents, who try to 
influence outcomes in a way that promotes their self interest. 
Halpern and Teague \citeyear{HT04} were perhaps the first to do this.
Their focus was on secret sharing and multiparty computation.
Following \shortcite{ADH13,AGLS14,BCZ12,GKTZ12}, we are interested in
applying these 
ideas to standard problems in game theory.
And like \shortcite{ADGH06,Alvisi05,BCZ12}, we are interested in what happens
when there is a mix of rational and faulty agents.  For the purposes
of this paper, we restrict to crash failures.  As we shall see, a
number of subtle issues arise even in this relatively simple setting.

We focus on the 
\emph{fair}
consensus problem,
where fairness means that the input of every agent is selected with equal probability.
Fairness seems critical in applications where we do not want
agents to be able to influence an outcome unduly.
For instance, when agents must decide whether to 
commit or abort a transaction, it is useful to ensure that
the outcome reflects the preferences of the agents, so that if 
a majority of agents prefers a particular outcome, it is selected with
higher probability.
Abraham, Dolev and Halpern \citeyear{ADH13} 
present a protocol for
fair 
leader election that even 
tolerates
coalitions of rational agents.
That is, in equilibrium, a
leader is elected, and
each agent is elected with equal probability.
Fair
leader election can be
used to solve 
fair
consensus (for example, once a leader is elected, the
consensus value can be taken to be the leader's value).
%
However, the protocol of \cite{ADH13} assumes that there are no
faulty agents.
Groce et al. \citeyear{GKTZ12} directly provide protocols for consensus with
rational agents, but again, they do not consider faulty agents and do
not require fairness.
Afek et al.~\citeyear{AGLS14}
and Bei, Chen, and Zhang~\citeyear{BCZ12}
provide 
protocols for consensus with crash failures and rational agents.
However, Afek et al.'s 
protocol works only under
strong
assumptions about  agents' preferences, such as
an agent having a strict preference for outcomes where it learns the
input of other agents,
while Bei, Chen, and Zhang  require that their protocol be robust to deviations
(that is, it achieves agreement even if rational agents deviate),
a requirement that we view as unreasonably strong (see
Section~\ref{sec:results}).
Neither of these protocols satisfy the fairness requirement.
Moreover, the protocol proposed by Afek et al. is not even
an equilibrium if some agent knows the input of other agents.
As we show, this is not an accident.

To explain our result, we need to briefly recall the standard notion
of \emph{ex post Nash equilibrium}.  In a setting where we have an
adversary, a protocol is an ex post equilibrium if no agent has any
incentive to deviate no matter what the adversary does.  Formally,
``no matter what the adversary does'' is captured by saying that even if
we fix the adversary's choice (so that the agents essentially know
what the adversary does), agents have no incentive to deviate.
Abraham, Dolev, and Halpern \citeyear{ADH13} provide protocols for
leader election (and hence consensus) that achieve ex post Nash
equilibrium if there are no failures.  Here, we show 
that even in synchronous systems, 
there
is no consensus protocol that is an ex post Nash equilibrium if there
can be even one crash failure.  

In the case of crash failures,  the adversary can be viewed as
choosing two things: the
\emph{failure pattern}---which agents fail, when they fail, and which
other agents they send a message to in the round that they fail, and 
the \emph{initial configuration}---what the initial preference of
each of the agents is. 
Roughly speaking, the reason that we cannot obtain an ex post Nash
equilibrium is that if the failure pattern and initial
configuration have a
specific form, a rational agent $i$ can take advantage of knowing this 
to increase the probability of
obtaining consensus on its preferred 
value.

There might seem to be an inconsistency here.  It is well known that
we can achieve consensus in synchronous systems with crash failures,
so it seems that we shouldn't have any difficulty dealing with 
one possibly faulty agent and one rational agent who does not follow
the protocol.  After all, we can view a
rational agent who deviates from the protocol as a faulty agent.
But there is no contradiction.  When the agent deviates from the
purported equilibrium, consensus is still reached, just on a different
value.   
That is, a rational agent may want to deviate so as to bias the
decision, although a consensus is still reached.

To get around our impossibility result, rather than trying to achieve
ex post Nash equilibrium, we assume that there is some distribution
$\pi$ on \emph{contexts}: pairs $(F,\vec{v})$ consisting of a failure
pattern $F$ and an initial configuration $\vec{v}$.  We show that 
under
appropriate
assumptions about $\pi$, if agents \emph{care only
  about consensus}---specifically, if (a) an agent's utility 
depends only on the consensus value achieved (and not, for example, on
the number of messages the agent sends) and (b) agents strictly prefer
reaching consensus to not reaching consensus---then there is a
Nash equilibrium that tolerates up to $f$ failures,  
as long as $f+1 < n$, 
where $n$ is the total number of agents.
Specifically, we make two assumptions about $\pi$, namely,
we assume that $\pi$
\emph{supports reachability} and is \emph{uniform}.
Roughly speaking, we say that $\pi$ supports reachability if it attributes small
probability to particular 
failure patterns that prevent information from one agent reaching an
agent that has not crashed by the end of the protocol;
we say that $\pi$ is uniform if it attributes equal probability to
equivalent failures of different agents. 
We believe that these assumptions apply in many practical systems; we
discuss this further in Section~\ref{sec:discussion}.

Our Nash equilibrium strategy relies on ``threats''; the threat that
there will be no consensus if an agent deviates (and is caught).
There might be some concern that these are empty threats, which will
never be carried out.  The notion of \emph{sequential
  equilibrium} \cite{KW82} is intended to deal with empty
threats.  Roughly speaking, a strategy is a sequential equilibrium if
all agents are best responding to what the others are doing even off
the equilibrium path.  We
generalize
sequential equilibrium to our
setting, where there might be failures, and show that 
the strategy that gives a Nash equilibrium can be slightly extended to
give a sequential equilibrium that tolerates up to $f$ failures.

The rest of this paper is organized as follows.  In
Section~\ref{sec:model} we discuss the model that we are using.  Our
main technical results on Nash equilibrium and sequential equilibrium are given 
in Section~\ref{sec:results}.
We conclude with
some discussion of the assumptions in
Section~\ref{sec:discussion}.

\section{Model}
\label{sec:model}
We consider a synchronous message-passing system with $n$ agents
and reliable communication 
channels between each pair of agents. Time is divided into
synchronous rounds.  
Each round is divided into a \emph{send phase}, where agents send
messages to other agents, a \emph{receive phase}, where
agents receive messages sent by other agents in the send
phase of that round, and an \emph{update phase}, where agents
update the value of variables based on what they have sent and received.
We denote by $N$ the set of agents and assume that they have
commonly-known identifiers in 
$\{0, \ldots, n-1\}$.
Round $m$ takes place between time $m$ and time $m+1$.

We now formalize the notion of run.  We take a \emph{round-$m$
  history} for agent 
$i$ to be a sequence of form 
$(v, t_1, \ldots, t_{m-1})$,
where $v$ is agent $i$'s initial preference and $t_j$ has the form
 $(s_j,r_j,d_j)$, where $s_j$ is the set of messages that $i$
sent in round $j$ tagged by who they were sent to, $r_j$ is the set of
messages that $i$ received in round
$j$
tagged by who they were sent by, and
$d_j \in \{\lambda\} \cup V$
is $i$'s decision (where $\lambda$
denotes that no decision has been made yet
and $V$ is the set of decision values).  
A \emph{global (round-$m$) history} has the form $(h_1, \ldots,
h_n)$ where $h_i$ is a round-$m$ history,
if $j$ receives a message ${\bf m}$ from $i$ in
round $m'$ of $h_j$, then $i$ sends ${\bf m}$ to $j$ in round $m'$ in
history $h_i$.  A \emph{run} $r$ is a function from 
time (which ranges over the natural numbers) to global
histories such that (a) $r(m)$ is a global round-$m$ history and (b)
if $m < m'$, then for each agent $i$, $i$'s history in $r(m)$ is a
prefix of $i$'s history in $r(m')$.

Agents are either correct or faulty in a run.
An agent fails only by crashing. If it crashes in round $m$ of run
$r$, then it may send a message to some subset of agents in round $m$,
but from then on, it sends no further messages.  We assume that all
messages sent are received in the round in which they are sent.
Thus, we take 
a \emph{failure} {\bf f} of agent $i$ to be a tuple $(i,m,A)$, where $m$ is a
round number (intuitively, the round at which $i$ crashes) and $A$ is
a set of agents (intuitively, the set of agents $j$ to whom $i$ can
send a message before it fails). 
We assume 
that
if $m >1$, then $A$
is non-empty,
so that $i$ 
sends a message to at
least one agent in round $m$ if $i$ fails in round $m$.
(Intuitively, if $m > 1$, we are identifying the failure pattern where
$i$ crashes in round $m$ and sends no message with the failure pattern
where $i$ crashes in round $m-1$ and sends messages to all the agents.)
A \emph{failure pattern} 
$F$
is a set of failures of distinct agents $i$.
A run $r$ \emph{has context $(F,\vec{v})$} if (a) $\vec{v}$
describes the initial preferences of the agents in $r$, (b) if 
$(i,m,A) \in F$, then $i$ sends all messages according to its protocol
in each round $m' < m$,  sends no messages in each round $m' > m$, and
sends messages according to its protocol only to the agents in $A$ in
round $m$, and (c) all messages sent in
$r$ are received in the round that they are sent. Let $\R(F,\vec{v})$
consist of all runs $r$ that have context $(F,\vec{v})$. Let $\R(F)$
consist of all runs that have $F$ as the set of failures.

In the consensus problem, we assume that each agent $i$ has an
initial preference $v_i$ in some set $V$.
For ease of exposition, we take $V = \{0,1\}$. (Our results
can easily be extended to deal with larger sets of possible values.)
A protocol achieves consensus if it satisfies the following
properties~\cite{FLP}: 
\begin{itemize}
  \item \textbf{Agreement}: No two correct agents decide different values.
  \item \textbf{Termination}: Every correct agent eventually decides.
  \item \textbf{Integrity}: All agents decide at most once.
  \item \textbf{Validity}: If an agent decides $v$, then $v$ was the
    initial preference of some agent. 
\end{itemize}

We are interested in one other property: fairness.  
Note that, once we fix a context, a protocol for the agents generates
a probability on \emph{runs},
and hence on outcomes, in the obvious way. 
Fairness just says that each agent has probability 
at least $1/n$ of having its
value be the consensus value, no matter what the context.
More precisely, we have the following condition:
\begin{itemize}
 \item \textbf{Fairness}: For each context $(F,\vec{v})$,
if $c$ of the nonfaulty agents in $F$ have
   initial preference $v$, then the probability of $v$ being the consensus
   decision 
conditional on $\R(F,\vec{v})$ 
is at least $c/n$.
\end{itemize}

It is straightforward to view a consensus problem as a game once we
associate a utility function $u_i$ with each agent $i$, where $u_i$
maps each outcome to a utility for $i$.  
Technically, it is an  \emph{extensive-form Bayesian game}.  
In a
Bayesian game, 
agents have \emph{types}, which encode private information.  In
consensus, an agent's type is its initial preference.  
A strategy for
agent $i$ in this game is just a protocol: a function from information
sets to actions. 
As usual, we view an extensive-form game as being
defined by a game tree, with the nodes where an agent $i$ moves into
\emph{information sets} where, 
intuitively, two nodes are in the same information set of 
agent $i$
if $i$ has the same information at both. In our setting, the nodes in a
game tree correspond to global histories, and agent $i$'s
information set
at a global history is determined by $i$'s history in that global
history; that is, we can take $i$'s information set at a global
history $h$ to consist of all global histories where $i$'s history is
the same as it is at $h$.  Thus, we identify an information set $I_i$ for
agent $i$ with a history $h_i$ for agent $i$.  If $I_i$ is the
information set associated with history $h_i$,
we denote by $\R(I_i)$ the set of runs $r$ where $i$ has history $h_i$
in $r(m)$.  

In game theory, a \emph{strategy} for agent $i$ is a function that
associates with each information set $I_i$ for agent $i$ a distribution
over the actions that $i$ can take at $I_i$.  In distributed computing,
a protocol for agent $i$ is a function that associates with each
history $h_i$ for agent $i$ a distribution over the actions that $i$
can take at $h_i$.  Since we are identifying histories for agent $i$
with information sets, it is clear that a protocol for agent $i$ can
be identified with a strategy for agent $i$.  
In consensus, the actions involve sending messages and deciding on values.
We assume that there is a special value $\bot$ that an agent
can decide on.  By deciding on $\bot$, an agent guarantees that there
is no consensus.  If we assume that an agent prefers to reach
consensus on some value to not reaching consensus at all, 
in
the language of Ben Porath \citeyear{Bp03},
this means that each agent has a \emph{punishment strategy}.

We next want to define an appropriate solution concept for our
setting. The standard approach is to say that an equilibrium is a
\emph{strategy profile} (i.e., a tuple of strategies, one for each agent)
where no agent can do better by deviating.
``Doing better'' is typically taken to mean ``gets a higher expected utility''.
However, if we do not have a probability on contexts, we
cannot compute an agent's expected utility.  
We thus consider two families of solution concepts.  In the first, 
we take ``doing better'' to mean that, for each fixed  context, no
agent can do better by deviating.  Once we fix the context, the
strategy profile generates a probability distribution on runs, and we
can compute the expected utility.  In the second  approach we assume a
distribution on contexts.

A strategy profile $\vsigma$
is an \emph{$\epsilon$--$f$-Nash equilibrium} if, for each fixed context 
$(F, \vec{v})$ where there are at most $f$ faulty agents in $F$, 
and all agents $i$, there is no
strategy $\sigma_i'$ for agent $i$ such that 
$i$ can improve its expected utility by more than
$\epsilon$. Formally,  
if $u_i(\vec{\tau} \mid
\R(F,\vec{v}))$ denotes $i$'s expected utility if strategy profile
$\vec{\tau}$ is played, conditional on the run
being in $\R(F,\vec{v})$, we require that for all 
strategies
$\sigma_i'$ for $i$, 
$u_i((\sigma_i',\vec{\sigma}_{-i}) \mid \R(F,\vec{v})) \le 
u_i(\vec{\sigma} \mid \R(F,\vec{v}))+\epsilon$.   
An $f$-Nash equilibrium is a $0$--$f$-Nash equilibrium.
The notion of
$f$-Nash equilibrium
extends the notion of ex post Nash equilibrium by
allowing 
up to $f$ faulty agents;
a $0$-Nash
equilibrium is an ex post Nash equilibrium.%
\footnote{This definition is in the spirit of the notion of
\emph{$(k,t)$}-robustness as defined by Abraham et al. \cite{ADGH06}, 
where coalitions of size $k$ are allowed in addition to $t$ ``faulty''
agents, but
here we restrict the behavior of the faulty agents to crash failures
rather than allowing the faulty agents to follow an arbitrary protocol,
and take $k=1$.
We also allow the deviating agents to fail (this assumption has no
impact on our results).}  

Given a distribution $\pi$ on contexts and a strategy profile $\vsigma$, 
$\pi$ and $\vsigma$ determine a
probability on runs denoted $\pi_{\vsigma}$ in the obvious way.
We say that $\vsigma$
is an \emph{$\epsilon$--$\pi$-Nash equilibrium} if, for all agents $i$
and all strategies $\sigma_i'$ for $i$, we have 
$u_i(\sigma_i',\vsigma_{-i}) \le u_i(\vsigma) +\epsilon$, where
now the expectation is taken with respect to
the probability $\pi_{\vsigma}$.
A $\pi$-Nash equilibrium is a 0--$\pi$-Nash equilibrium.
If $\pi$ puts probability 1 on there being no failures, then
we get the standard notion of ($\epsilon$-) Nash equilibrium.

\section{Possibility and Impossibility Results for Consensus}\label{sec:results}
In this section, we consider the consensus problem from a
game-theoretic viewpoint.
We focus on the case where agents care only about consensus, since
this type of utility function seems to capture many situations of interest.  
For the rest of this section, let $\beta_{0i}$ be $i$'s
utility if its initial preference is decided, let $\beta_{1i}$ be $i$'s
utility if there is consensus, but not on $i$'s initial preference, and let
$\beta_{2i}$ be $i$'s utility if there is no consensus.  The
assumption that agents 
care 
only about consensus means that, for all $i$,
$\beta_{0i} > \beta_{1i} > \beta_{2i}$. 

Note that although we assume that agents prefer consensus to no
consensus, unlike Bei, Chen, and Zhang.~\citeyear{BCZ12}, we do not require that
our algorithms guarantee consensus
when rational agents deviate.
Our algorithm does guarantee that
there will be consensus if there are no 
deviations.
On the other hand,
we allow for the possibility that a deviation by a rational agent will
result in there being no consensus. For example, suppose that a
rational agent pretends to fail in a setting where there is a bound
$f$ on the number of crash failures. That means that if $f$ other agents
actually do crash, then some agent will detect that $f+1$ agents seem
to have crashed.   Our algorithm requires that if an agent detects
such an inconsistency, then it aborts.  If the probability that $f$
agents actually crash is low, in our framework, a rational agent may
decide that it is worth the risk of pretending to crash if the
potential gain is sufficiently large.  Bei, Chen, and Zhang would not permit
this, since they require consensus even if rational agents deviate
from the algorithm.  This requirement thus severely limits the
possible deviations.

\subsection{An Impossibility Result}

We start by showing that there is no 
fair consensus protocol
that is an $f$-Nash equilibrium.

\begin{theorem}
\label{theo:imp-fair}
If $\vsigma$ solves fair consensus,
agents care only about consensus, and $f \ge 1$, then $\vsigma$ is not 
an $f$-Nash equilibrium
\end{theorem}
\begin{proof}
Consider the initial configuration $\vec{v}$ where all agents but $i$
have initial preference 0 and $i$ has initial preference 1.
If $F^1$ 
is the failure pattern where no agent fails, 
by Fairness, the agents must decide 1 with positive probability in
context $(\vec{v},F^1)$.  
It follows that there must be a failure pattern $F^2$ where only agent $i$
fails but the agents  decide 1 with positive probability in
context $(\vec{v},F^2)$. (In $F^2$, $i$ fails only after a decision has
been made in $F^1$.)
If $F^0$ is the failure pattern where only $i$ fails, and $i$ fails
immediately, before sending any messages, then it is clear that no
agents can distinguish this context from one where all agents have
initial preference 0, so all agents must decide 0, by the Validity
requirement.

Put a partial order $\le$ on failure patterns where only $i$ crashes by taking
$F \le F'$ if either $i$ crashes in an 
earlier round in $F$ than in $F'$, or $i$ crashes in the same round
$m$ in both $F$ and $F'$, but the set of agents to whom $i$ sends a
message in $F$ is a subset of the set of agents to whom $i$ sends a
message in $F'$.
Clearly $F^0 < F^2$.  
Thus, there exists a minimal failure pattern $F^*$ such that $F^0 <
F^* \le F^2$, only $i$ fails in $F^*$,
the consensus is on 1 with positive probability in context
$(\vec{v},F^*)$,  the consensus is 0 with probability 1 in all contexts
$(F,\vec{v})$
where only agent $i$ fails in $F$ and $F < F^*$.  
We can assume without loss of generality that $i$ sends a message to
some agent $j$ in the round $m$ in which $i$ fails.
To see this, note that if $i$ crashes in the first
round then $i$ must send a message to some agent (otherwise $F^* =
F^0$ and the decision is 0 with probability 1).
And if $i$ crashes in round $m > 1$, we have assumed that $i$ sends at
least one message before crashing (recall that we identify an agent
crashing at round $m > 1$ and sending no messages with the agent
crashing at round $m-1$ and sending to all agents).

Now suppose that an agent $j$ that receives a message from $i$ in
round $m$ pretends not to receive that message.
This makes the situation indistinguishable from
the context $(F,\vec{v})$ where $F$ is just like $F^*$ except that $i$
does not send a message to $j$ in round $m$.  Since 
$F^0 \le F < F^*$, 
the decision must
be 0 with probability 1 in context $(\vec{v},F)$.
Since $j$ has initial preference $0$ in $\vec{v}$, $j$ can increase its
expected utility by this pretense, so $\vec{\sigma}$ is not 
an $f$-Nash equilibrium.
\end{proof}

\subsection{Obtaining a $\pi$-Nash equilibrium}

We now prove a positive result.  If we are willing 
to assume
that there is a 
distribution $\pi$ on contexts with some reasonable properties, then
we can get a 
fair
$\pi$-Nash equilibrium. But, as we show below, there are
some subtle problems in doing this.

Before discussing these problems, it is useful to recall some results
from social choice theory. Consider a setting with $n$ agents where
each has a preference order (i.e., a total order) over some set $O$ of
outcomes.
A \emph{social-choice function}
is a (possibly randomized) function that maps a profile of preference orders
to an outcome. For example, we can consider
agents trying to elect a leader, where each agent has a preference order
over the candidates; the social-choice function chooses a leader as a
function of the expressed preferences. A
social-choice function is \emph{incentive compatible} if no agent can
do better by lying about its preferences.  The
well-known \emph{Gibbard-Satterthwaite} theorem
\cite{Gibbard73,Satter75} says that if there are at least three
possible outcomes, then the only incentive-compatible deterministic social-choice
function $f$ is a \emph{dictatorship}; i.e., the function $f$ just
chooses a player $i$ and takes the outcome to be $i$'s most-preferred
candidate, ignoring all other agents' preferences.  Gibbard
\citeyear{Gibbard77} extends this result to show that if there are at
least three outcomes, then the only randomized incentive-compatible
social-choice function 
is a \emph{random dictatorship}, which essentially amounts to choosing
some player $i$ according to some probability distribution and then
choosing $i$'s value. 

Bei, Chen, and Zehang \citeyear{BCZ12} point out that a strategy
profile that solves consensus can be viewed as a social-choice
function: agents have preferences over three outcomes, 0, 1,
and $\bot$, and the consensus value (or $\bot$, if there is no
consensus) can be viewed as the outcome chosen by the function.    A 
strategy profile that is a Nash equilibrium is clearly incentive-compatible;
no agent has an incentive to lie about its preferences.  Thus, it
follows from Gibbard's \citeyear{Gibbard77} result 
that a solution to rational consensus must
be a randomized
dictatorship. And, indeed, our protocols can be viewed as
implementing a randomized dictatorship: one agent is chosen at random,
and its value becomes the consensus value.  However, implementing such a
randomized dictatorship in our
setting is nontrivial because of the possibility of failures.%
\footnote{We remark that Theorem 1 of Bei, Chen, and Zhang
  \citeyear{BCZ12} claims  that, given a fixed failure pattern, a
  strategy profile for  consensus that is a Nash equilibrium must
  implement a dictatorship, rather than randomized dictatorship.
  While this is true if we restrict to deterministic strategies,
  neither we nor Bei, Chen, and Zhang do so.  We have not checked
  carefully whether results of Bei, Chen, and Zhang that depend on
  their Theorem 1 continue to hold once we allow for randomized
  dictatorships.}

\fullv{\subsubsection{A naive protocol}}
  \shortv{\paragraph{\rm {\bf A naive protocol}}}

We start with a protocol that, while not solving the problem, has many
of the essential features of our solution, and also helps to point out
the subtleties.
Consider the following slight variant of one of the early protocols
for consensus \cite{DS1}: In round 1, each agent $i$ broadcasts a tuple
$(i, v_i,x_{i0}, \ldots, x_{if})$, where $v_i$ is $i$'s initial preference, and
$x_{it}$ is a random element in $\{0, \ldots, n-t\}$.  For round
$2, \ldots, f+1$, each agent $i$ broadcasts all the tuples $(j,v_j,\vec{x}_j)$
that $i$ received and did not already forward in earlier rounds.  
At the end of round $f+1$, each agent checks for consistency;
specifically, it checks 
that it has received tuples from at least $n-f$ agents and that 
it has not received distinct tuples claimed to have been sent by
some agent $j$.  If $i$  
detects an inconsistency, then $i$ decides $\bot$.  Otherwise,
suppose that $i$ received tuples from $n-t$ agents. Then $i$ computes the
sum mod $n-t$ of the values $x_{jt}$ for each agent $j$ from which it
received a tuple.  If the sum is $S$, then $i$ decides on the value 
of the agent with the $(S+1)$st highest id among the $n-t$ agents from
which it received tuples.
(Here is where we are implementing the random dictatorship.)
Note that the random value $x_{jt}$ is used by $i$ in
computing the consensus value if exactly $t$ faulty agents are discovered;
the remaining random values sent by agent $j$ in the first round are
discarded.    

It is straightforward to check that if all nonfaulty agents follow
this protocol, then they will all agree on the set of tuples received
(see the proof of Theorem~\ref{theorem:corr} for an argument similar
in spirit),  
and so will choose the same decision value, and each agent whose value
is considered has an equal chance of having their value determine the
outcome.  But this will not be in general a $\pi$-Nash equilibrium if
$\pi$ \emph{allows up to $f$ failures}, that is, $\pi$ puts 
probability 0 on all failure patterns that have more than $f$
failures and $f \ge 2$.

Consider a distribution $\pi$
that puts positive probability on all contexts with at most $f$ failures,
and
an initial configuration where
agent 1 prefers 1, but all other agents prefer 0.
Agent 1 follows the protocol in the first round,
and receives a message from all the other agents. We claim that agent 1
may have an incentive to pretend to fail (without sending any
messages) at this point.  Agent 1 can gain by doing this if one of the
other agents, say agent 2, crashed in the first round and sent a message only to
agent 1.  In this case, if 1 pretends to crash, no other agent will
learn 2's initial preference, so 1's initial preference will have a somewhat
higher probability 
(at least $\frac{1}{n-1} - \frac{1}{n}$) 
of becoming the
consensus decision.  Of course, there is a risk in pretending to
crash: if $f$ agents really do crash, then an inconsistency will be
detected, and the decision will be $\bot$.  
Let $\alpha_{<f}$ be the probability
of there being fewer than $f$ failures and at least one 
agent crashing in the first round who does not send to any agent other
than 1 (this is the probability that 1 gains some utility by its
action); let $\alpha_{=f}$ be the probability of there being $f$ crashes
other than 1 (this is an upper bound on the probability that 1 loses
utility by its action).  Then 1's expected gain by deviating is at
least 
$$(\beta_{0i}-\beta_{1i})\left(\frac{1}{n-1} - \frac{1}{n}\right)\alpha_{<f} - (\beta_{0i} - \beta_{2i})\alpha_{=f}.$$

This is a small quantity. However, if $f$ is reasonably large and
failures are unlikely, we would expect $\alpha_{=f}$ to be much smaller
than $\alpha_{<f}$, so as the number $f$ of failures that the
protocol is designed to handle increases, deviating becomes more and
more likely to produce a (small) gain.

\fullv{\subsubsection{A $\pi$-Nash equilibrium}}
\shortv{\paragraph{\rm {\bf A $\pi$-Nash equilibrium}}}
There are 
three
problems with the preceding protocol.  The first is
that, even if 1 pretends to 
fail, 1's value will be considered a potential consensus value,
since everyone received the value before 1 failed.  
This means that there is little downside in pretending to fail.
Roughly speaking,
we deal with this problem by taking into consideration only the values
of nonfaulty agents when deciding on a consensus value.
The 
second
problem is that since agents learn the random values
$(x_{i0}, \ldots, x_{if})$ that will be used in determining the consensus 
value in round 1, they may be able to guess with high probability 
the value that will be decided on at a point when they can still
influence the outcome.  To address this problem, agents do not send
these random values in the first round; instead, they 
use secret sharing \cite{shamir}, so as to allow the 
nonfaulty agents to reconstruct these random values
when they need to decide on the consensus value.  This prevents agents
from being able to guess with high probability what the decision will
be too early. 
The 
third
problem is that in some cases agents can safely lie about
the messages sent by other agents
(e.g., $i$ can pretend that another agent did not crash).
We could solve this by assuming
that messages can be signed using unforgeable signatures.  We do not
need this or any other cryptographic assumption.  Instead, we use some
randomization to ensure 
that if an agent lies about a message that was sent, it will be caught
with high probability.

Thus, in our algorithm, an agent $i$ generates random numbers for two
reasons.  The first is that it generates $f+1$ random numbers $(x_{i0},
\ldots, x_{if})$, where $x_{it}$ is used in choosing the consensus
value if there are exactly $t$ faulty agents discovered, and then, as
we suggested above, shares them using secret sharing,
so that the numbers can be reconstructed at the appropriate time (see below). 
The second is that it
generates $n-1$ additional random numbers,
denoted $z_{ij}^m[i]$, one for each agent $j \ne i$,
in each 
round $m$,
and sends 
them
to $j$ in round $m$. 
Then if agent $j$
claims that it got a message in round $m$ from $i$, it will have to 
also provide $z_{ij}^m[i]$ as proof.  

In more detail, we proceed as follows.  Initially, each agent $i$ generates
a random tuple $(x_{i0}, \ldots, x_{if})$, where 
$x_{it}$ is in $\{0, \ldots, n-t\}$.
  It then computes $f+1$  
random polynomials $q_{i0}, \ldots, q_{if}$, each of degree 
$1$,
such that
$q_{it}(0) = x_{it}$.  It then sends $(q_{i0}(j), \ldots, q_{if}(j))$
to agent $j$. The upshot of this is that 
no agent
will be able to compute $x_{it}$ given this information 
(since one point
on a degree-1 polynomial $q_{it}$ gives no information
regarding $q_{it}(0)$). In addition, in round 
1, each agent $i$
sends $v_i$ to each agent $j$, just as in the naive algorithm;
it also generates the random number $z_{ij}^1[i]$ and a special random
number $z$, and sends each agent the vector $z_{ij}^1 = (z_{ij}^1[1],
\ldots, z_{ij}^1[n])$, where 
$z_{ij}^1[j'] = 0$ for $j' \ne i,j$
and $z_{ij}^1[j] = z$.
(As we said, these random numbers form a ``signature''; their role
will become clearer in the proof.)
Finally, in round 1, agent $i$ sends a status report $SR^1_i$; we
discuss this in more detail below.
In the receive phase of round 1, agent $i$ adds all the values 
received from other agents to the set $ST_i$.  

In round $m$ with $2 \le m \le f$, $i$ again sends 
a status report $SR_i^m$
and a vector $z_{ij}^m$.
For each agent $j$, $SR_i^m[j]$ is a tuple of the form $(m,x)$,
where $m$ is the first round that $i$ knows that $j$ crashed ($m =
\infty$ if $i$ believes that $j$ has not yet crashed),
and $x$ is either the vector $z_{ji}^{m-1}$ of random values sent by $j$ in $m-1$ (if $i$ believes that $j$ has not yet crashed)
or an agent that told $i$ that $j$ crashed in round $m$.
The tuple $z_{ij}^m$ is computed by setting $z_{ij}^m[l]$ for $l \ne i,j$ to be
$z_{li}^{m-1}[l]$, the random 
number
sent by $l$ in the previous round
(this will be used to prove that $i$ really got a message from $l$ in
the previous round---it is our replacement for unforgeable
signatures); again, $z_{ij}^m[i]$ is a random value 
generated by $i$.
In round $f+1$, $i$ also sends $j$ the secret shares $y_{li}^t$ it
received in round 1 from each agent $l$ (i.e., the value $q_l^t(i)$
that it received from $l$, assuming that $l$ did not lie).  
This enables
$j$ to compute the polynomials $q_{it}$, and hence the secret
$q_{it}(0) = x_{it}$ for $0 \le t \le f$. 

If $i$ detects an inconsistency in round $m \le f+1$, then
$i$ decides $\bot$, where
$i$ \emph{detects an inconsistency in round $m$} if 
the messages received by $i$ are inconsistent with all agents
following the protocol 
except that up to $f$ agents may crash.
\shortv{(In the full paper
\cite{HV16}, we
  give an exhaustive list of all the ways 
  that $i$ can detect an inconsistency.)}\fullv{This can happen if 
\begin{enumerate}
  \item $j$ sends incorrectly formatted messages;

\item $m=2$ and
agents $j'$ and $j'' \ne i$ disagree about the random values $z_{jj'}^1[j']$ and $z_{jj''}^1[j'']$ sent by $j$
in round $1$;

\item $m>2$ and some agent $j' \neq j$ reports that $j$ sent 
a value $z_{jj'}^{m-1}[i]$ in round $m-1$ 
different from the value $z_{ij}^{m-2}[i]$ sent by $i$ to $j$ in round
$m-2$; 

\item $m=f+1$ and it is not possible to interpolate a polynomial
  $q_j^t$ through the shares $y_{ji}^t$ received by $i$ from $j$ in 
round 1 and
  the values $y_{jl}^t$ received 
  from $l \ne j$ in round $f+1$.
\item some agent $j'$ sends $i$ a status report in round $m$ that says that
$j$ crashed in some round $m'$ and either $i$ receives a message
  from $j$ in round $m'' > m'$ or some agent $j''$ sends $i$ a status
    report saying that it received a message from $j$ in a round $m'' > m'$;
\item for some agents $j$, $j'$, and $j''$,
$j$ sends $i$ a status report in round $m$ that says
that $j''$ crashed in round $m'$ and that $j'$ reported this,
but $j'$ sends $i$ a status report in round $m$ that says that 
$j''$ did not crash before round $m'' > m'$;
\item 
for some agents $j$, $j'$, and $j''$, $j$ sends $i$ a status report in
round $m$ that says that $j'$ did not
crash by round $m-1$ and $j''$ crashed in some round $m' < m$, while 
$j'$ sends $i$ a status report in round $m-1$ saying 
that $j''$ crashed in round $m'' < m'$
(so either $j$ ignored the report about $j''$ sent by $j'$ or $j'$
lied to $j$);
\item more than $f$ crashes are detected by $i$ by round $m$ (i.e.,
  $f$ or more agents have not sent 
  messages to $i$ 
 or were reported to crash 
  in some round up to
  and including $m$).
\end{enumerate}
}
If agent $i$ does not detect an inconsistency at some round $m \le
f+1$, $i$ proceeds as follows in round $f+1$.
For each round $1 \le m \le f+1$ in a run $r$, agent
$i$ computes $NC_m(r)$, the set of agents that it believes
did not crash up to and including round $m$.
Take $NC_0(r) = N$ (the set of all agents).   
Say that round $m$ 
in run $r$ 
\emph{seems clean} if $NC_{m-1}(r) = NC_{m}(r)$.
As we show (Theorem~\ref{theorem:corr}), if no inconsistency is 
detected in run $r$, then  
there must be a round in $r$ that seems clean. Moreover, we show that if
$m^*$ is the first round 
in $r$
that seems clean to a nonfaulty agent $i$, then all the nonfaulty
agents agree that $m^*$ is the first round that seems clean in $r$, and they
agree on 
the initial preference of all agents in $NC_{m^*}(r)$, and the random numbers
sent by these agents in round $1$ messages
in run $r$.
The agents then use these random numbers to choose an agent $j$
among the agents in $NC_{m^*}(r)$ and take $v_j$ to be the consensus value.

The pseudocode for the strategy (protocol) $\vsigma^{\osc}$ that
implements this idea is given in Figure~\ref{fig1}.
\shortv{We discuss the algorithm in more detail in the full paper.}


\fullv{
Lines~\ref{line:initbegin}--\ref{line:initend} initialize the values of
$ST$ and
$SR^1[j]$,
as well as the random numbers required in 
round $1$; that is,
$i$ generates $x_i[t]$ and the corresponding polynomial $q_i^{t}$ used
for secret sharing for 
$0 \le t \le f$, 
and random vectors $(z_{ij}^1[1], \ldots, z_{ij}^1[n])$ for $j \ne i$,
where $z_{ij}^1[l] \in \{0, \ldots, n-1\}$.
In phase $1$ (the ``sending'' phase) of round $m$,
$i$ sends $SR_i^m$
and 
$z_{ij}^m$.
If $m=1$, then $i$ also sends $v_i$
and $(y_{ij}^0, \ldots, y_{ij}^f)$ to $j$,
where $y_{ij}^t = q_i^t(j)$; that is, $y_{ij}^t$ is $j$'s share of the
secret $x_i^t$. 
Finally, if $m=f+1$, instead of sending $z_{ij}^m$ to $j$,
$i$ sends all the shares $y_{li}^t$ it has received from other agents,
so that all agents can compute the secret
(lines~\ref{line:phase1}-\ref{line:phase1:end}).   
In phase 2  (the ``receive'' phase)
of round $m$, $i$ processes all the messages received and
keeps track of all agents who have crashed
(lines~\ref{line:phase2}-\ref{line:phase2:end}). 
If $i$ receives a round $m$ message from $j$, then $i$ 
adds $(j,v_j)$ to $ST_i$ if $m=1$,
includes in $SR_i^m[j]$ the vector $z$ sent by $j$ to $i$, and updates the
status report $SR_i^m[l]$ of each agent $l$. Specifically, 
if $j$ reports that $j'$ crashed in a round $m'$ and 
$i$ earlier considered it possible that $j'$ was still nonfaulty at
round $m'$, then
$i$ includes in $SR_i^m[l]$ the fact that $j'$ crashed and 
that $j$ is an agent that reported this fact
(lines~\ref{line:update-status}-\ref{line:update-status:end}); 
if $i$ does not receive a round $m$ message from $j$ and $i$ believed that $j$ did not crash before,
then $i$ marks $j$ as crashed (line~\ref{line:register-crash}). 
In phase 3 (the ``update'' phase) of round $m \le f$, $i$ 
generates the random value $z^{m+1}_{ij}[i]$ for the next round.
If $i$ detects an inconsistency, then $i$ decides $\bot$
(line~\ref{line:punish}); if no inconsistency is
detected by the end of round $f+1$, then $i$ decides on a value
(lines~\ref{line:decide}-\ref{line:decide:end})
by computing the set $NC_{m'}$ for every round $m'$, determining
the earliest round $m^*$ that seems clean ($NC_{m^*} = NC_{m^*-1}$), 
computing a random number $S \in \{0, \ldots, n-t-1\}$, where $t$ 
is the number of crashes that occurred before $m^*$, by summing the
random numbers $x_j[t]$ of $j \in NC_{m^*}$ (computed by
interpolating the polynomials), 
and deciding on the
value of the agent in $NC_{m^*}$ with the $(S+1)$st highest id.
}

\fullv{\begin{figure}[t]}
\shortv{\begin{figure}[e]}
\begin{algorithm}[H]
\caption{$\sigma_i^{\osc}(v_i)$: $i$'s consensus protocol with initial value $v_i$}
\label{fig1}
{
\scriptsize
\begin{algorithmic}[1]
\State{decided $\gets$ $\false$}\label{line:initbegin}
\State{$ST_i$ $\gets$ $\{(i,v_i)\}$}
\State{$z$ $\gets$ random in $\{ 0\, \ldots \, n-1\}$}
\ForAll{$j \neq i$}
	\State{$SR_i^1[j]$ $\gets$ $(\infty,\bot)$}\Comment{All agents are initially active}
	\ForAll{$l \neq j,i$}
	
		\State{$z_{ij}^1[l]$ $\gets$ $0$}
	\EndFor
	\State{$z^1_{ij}[i]$ $\gets$ random in $\{ 0\, \ldots \, n-1\}$}\Comment{Random number to be used by $j$ in round $2$}
	\State{$z^1_{ij}[j]$ $\gets$ $z$}\Comment{Proves that $i$ sends round 1 message to $j$}
\EndFor

\ForAll{$0 \leq t \leq f$}
	\State{$x_{i}[t]$ $\gets$ random in
          $\{0,\ldots,n-t-1\}$}\Comment{A random number for each
          possible value of $t$} 
	  \State{$q_i^{t}$ $\gets$ random polynomial of degree $1$ with $q_i^{t}(0) = x_i[t]$} 
	\ForAll{$j \ne i$}
		\State{$y_{ij}[t]$ $\gets$ $q_i^{t}(j)$}
	\EndFor
\EndFor \label{line:initend}

\Statex

\ForAll{round $1 \leq m \leq f+1$ such that $\neg \mbox{decided}$}

\Phase{1}{send phase}\label{line:phase1}
	\ForAll{$j \neq i$}
		\If{$m = 1$}
		Send $\langle v_i,SR_i^m,(y_{ij}^0,\ldots,
                y_{ij}^f),z^m_{ij}\rangle$ to $j$ 	    	\EndIf
	  	\If{$2 \le m \le f$}
	    	 	Send $\langle SR_i^m,z^m_{ij}\rangle$ to $j$   	\EndIf
                                                \If{$m=f+1$}
	   		Send $\langle SR_i^m,(y^0_{li}, \ldots,
			     y^f_{li})_{l \neq j}\rangle$ to $j$ \EndIf
	\EndFor
\EndPhase\label{line:phase1:end}

\Statex

\Phase{2}{receive phase}\label{line:phase2}
	\State{$SR_i^{m+1}$ $\gets$ $SR_i^m$}
	\ForAll{$j \neq i$}
		\If{receive valid message from $j$}
			\If{$m=1$}
				$ST_i$ $\gets$ $ST_i  \cup \{(j,v_j)\}$\label{line:msg}\Comment{$ST_i$ contains all the values that $i$ has seen}
			\EndIf
			\State{$SR_i^{m+1}[j] \gets (\infty,z^m_{ji})$}\Comment{Note that $j$ is still active}
             \ForAll{$l \ne i,j$}
				\If{$SR_j^m[l] = (m',j')$ and $SR_i^m[l] = (m'',j'')$ and $m' < m''$}\label{line:update-status}
					\State{$SR_i^{m+1}[l] \gets (m',j)$}
                                        \Comment{$l$   crashed earlier than previously thought}
					\State{$z_{il}^{m+1}[j] \gets \bot$}
			    \ElsIf{$SR_j^m[l] = (\infty,z_j^m)$} 
					\State{$z_{il}^{m+1}[j] = z_{ji}^m[j]$}
				\EndIf\label{line:update-status:end}
			\EndFor\label{line:msg:end}
	    \ElsIf{$SR_i^{m+1}[j] = (\infty,z')$ for some $z'$}
 			\State{$SR_i^{m+1}[j] \gets	(m,i)$}\label{line:register-crash}\Comment{$i$ detects a crash of $j$}
             \ForAll{$l \ne i$}
				\State{$z_{il}^{m+1}[j] \gets \bot$}
             \EndFor
		\EndIf
	\EndFor
\EndPhase\label{line:phase2:end}

\Statex
	
\Phase{3}{update phase}\label{line:phase3}
	\If{an inconsistency is detected}
			\State{\textbf{Decide}($\bot$)}\label{line:punish}\Comment{Punishment}
			\State{decided $\gets \true$}
  	\ElsIf{$m \le f$}
		\ForAll{$j \neq i$}
			\State{$z^{m+1}_{ij}[i]$ $\gets$ random in $\{0,\ldots,n-1\}$}
		\EndFor
	\ElsIf{decided = $\false$} \label{line:decide}
\State{$NC_0 = N$}
		\ForAll{$1 \leq m' \leq f+1$}
			\State{$NC_{m'}$ $\gets$ $\{j \in N -\{i\}	 \mid \forall m'' \le m',
			 l (SR_i^{f+2}[j] \neq  (m'',l)) \cup \{i\}$}\Comment{Agents that did not crash up to round $m'$} 
		\EndFor	
		\State{$m^*$ $\gets$ first round $m'$ such that $NC_{m'} = NC_{m'-1}$}\Comment{First round that seems clean}
		\State{$t$ $\gets$ $n - |NC_{m^*}|$}\Comment{Number of crashes prior to $m^*$}
		\ForAll{$j \in NC_{m^*}$}
			\State{$q_j^{t}$ $\gets$ unique polynomial interpolating 
			the values $y_{jl}^t$ received}\Comment{otherwise, an inconsistency was detected}
			\State{$x_j[t] \gets q_j^{t}(0)$}
		\EndFor
		\State{$S \gets \sum_{j \in NC_{m^*}}x_{j}[t] \mbox{ mod  } (n-t)$}\Comment{Calculate a random number in $0,\ldots,n-t-1$}
		\State{\textbf{Decide}($v_j$), where $j$ is the $(S+1)$st highest id in $NC_{m^*}$}
	\EndIf\label{line:decide:end}
\EndPhase \label{line:phase3:end}
\EndFor
\end{algorithmic}
}
\end{algorithm}
\end{figure}

We now prove that $\vsigma^{\osc}$ gives a $\pi$-Nash
equilibrium, under reasonable assumptions about $\pi$. We first prove
that the protocol satisfies all the properties of fair consensus
without making any assumptions about $\pi$.

\begin{theorem}
\label{theorem:corr}
$\vsigma^{\osc}$ solves fair consensus if at most $f$
agents crash, 
$f+1 < n$, 
and all the remaining agents follow the protocol.
\end{theorem}
\shortv{
The proof of this theorem (and all others) is in the full paper.
}
\fullv{
\begin{proof}
Consider a run $r$ where all agents follow $\vsigma^{\osc}$  and
at most $f$ agents crash.  It is easy to see that no inconsistency is
detected in $r$.
Since an agent crashes in at most one round and there are at most $f$
faulty agents, 
there must exist a round $1 \le m \leq f+1$ when no agent crashes. Let
$m^*$ be the first such round.
We prove that for all nonfaulty agents $i$ and $j$, $NC_m^i(r) =
NC_m^j(r)$ for all $m \leq m^*$ (where $NC_m^i(r)$ denotes $i$'s
version of $NC_m(r)$ in run $r$, and similarly for $j$). To see this,
fix two nonfaulty agents $i$ and $j$. Agent
$i$ adds agent $l$ to $NC_m^i(r)$ iff $i$ receives a message from $l$ in
every round $m' \le m$ 
of run $r$,
and $i$ receives no status report indicating that $l$ crashed in some
round $m' \le m$.
If $m < m^*$, then it must be the case that $j$ also received a
message from $l$ in every round $m' < m$ of $r$
and neither received nor sent a status
report indicating that $l$ crashed in a round $m' \le m$;
otherwise
$j$ would have learned about this crash by round $m$
and would have told $i$ by round $f+1$ that $l$ was faulty (since $j$
is nonfaulty). 
Thus, $l \in NC_m^j(r)$.
If $m=m^*$, then $l$ sends a round $m'$ message 
to all agents for all $m' < m^*$; and since no
agents fail in round $m^*$, by assumption,
we again have $l \in NC_m^j(r)$.
Thus, $NC_m^i(r) \subseteq NC_m^j(r)$; similar arguments give the
opposite inclusion.

Note that since no agent crashes in round $m^*$, it is easy to see
that we must have $NC^i_{m^*}(r) = NC^i_{m^*-1}(r)$ for all nonfaulty
agents $i$,
so round $m^*$ 
seems
clean.
With these observations, we can now prove that $\vsigma^{\osc}$ satisfies
each requirement of Fair Consensus 
in $r$.

\textbf{Validity}: Since no inconsistency is detected, every agent $i$
decides a value different from $\bot$ in $r$. 
Agent
$i$ always finds some round $m^*$ that seems clean, computes a nonempty set $NC_{m^*}(r)$, which includes at least $i$,
and knows the random numbers sent by these agents in round $m^*$.
Since $ST_i$ contains only initial preferences, $i$ decides the initial preference of some agent in $NC_{m^*}(r)$.

\textbf{Termination and Integrity}: Every agent either crashes before deciding
or decides exactly once at the end of round $f+1$.

\textbf{Agreement}: We have shown that all nonfaulty
agents $i$ and $j$ agree on $NC_{m}(r)$ for all $m \leq m^*$.  
We thus omit the superscripts $i$ and $j$ on $NC_m(r)$ from here on in.
Given this, they agree on whether each round $m \leq m^*$ seems clean
and thus agree that some $\overline{m} \leq m^*$ is the first round that
seems clean in $r$. 
Moreover, $i$ and $j$ receive identical round 
$1$ messages
from the agents in $NC_{\overline{m}}(r)$.
It follows that $i$ adds a tuple $(l,v_l)$ to $ST_i$ for $l \in
NC_{\overline{m}}(r)$ iff $j$ adds that tuple to $ST_j$.
Suppose that $|NC_{\overline{m}}(r)| = n-t$.  Since $NC_{\overline{m}}(r)$
must include all the nonfaulty agents, we must have $t \le f$.
Clearly, if  $l \in NC_{\overline{m}}(r)$, then $i$ and $j$ 
must receive the values $y_{li}^t$ and $y_{lj}^t$ 
in round $1$ messages sent by $l$.
Agents $i$ and $j$ also receive $y_{ll'}^t$ from each nonfaulty agent
$l'$. Since there are at least 
$n-f \ge 2$ nonfaulty agents, 
and $l$ follows $\sigma_l^{\osc}$, $i$ and $j$ will be able to
interpolate the polynomial $q_l^t$, and compute $x_l[t] = q_l^t(0)$.
Consequently, $i$ and $j$ agree on the information relevant to the
consensus decision, so must decide on the same value.

\textbf{Fairness}: The probability of the initial preference of each agent
in $NC_{\overline{m}}(r)$ being decided is $1/|NC_{\overline{m}}(r)|$.
Since $|NC_{\overline{m}}(r)| \leq n$, if $c$ nonfaulty agents in $NC_{\overline{m}}$ initially have preference $v$,
then the probability of $v$ being decided is at least
$c/|NC_{\overline{m}(r)}| \geq c/n$. 
Since $NC_{\overline{m}}(r)$ contains all the nonfaulty agents, Fairness holds.
\end{proof}
}

It remains to show that $\vsigma^{\osc}$ is a $\pi$-Nash equilibrium.
We show that $\vsigma^{\osc}$ is a $\pi$-Nash equilibrium under
appropriate assumptions about $\pi$. Specifically, we assume that $\pi$ 
\emph{supports reachability} and is \emph{uniform}, notions that we now define.
The reachability assumption has three parts. The first two parts consider
how likely it is that some
information that an agent $j$ has will reach an agent that will decide
on a value; the third part is quite similar, and considers how
likely it is that a nonfaulty agent becomes aware that an agent $j$
failed in round $m$. Of course, the answer to these questions depends
in part on whether agents are supposed to send messages in every round
(as is the case with $\vsigma^{\osc}$).  In the formal definition, we
implicitly assume that this is the case.  (So, effectively, the
reachability assumption is appropriate only for protocols where
agents send messages in every round.)  
Given 
agents $i$ and $j \neq i$,
a round-$m$ information set $I_i$ for $i$, 
a failure pattern $F$   
compatible with $I_i$, in that $\R(F) \cap \R(I_i) \ne \emptyset$, 
and $m' \ge m$, say that \emph{a nonfaulty agent
$l \ne i$
  is reachable from $j$
  without $i$
  between rounds $m'$ and $f+1$
  given $F$} 
if there is a sequence $j_{m'}, \ldots, j_{f+1}$ of agents
different from $i$
such that $j=j_{m'}$, for $m'' 
  = \{m', \ldots, f\}$,  
$j_{m''}$
has not failed prior to
  round $m''$ 
  according to $F$, and either 
does not fail in round $m''$
or, if $m'' <
  f+1$, $j_{m''}$ fails in round $m''$ but 
  sends a message to $j_{m''+1}$ before failing (i.e., if 
   $(j_{m''},m'',A) \in F$, then $j_{m''+1} \in A$),
  and $l = j_{f+1}$.

Note that if $j$ is nonfaulty according to $F$, then a nonfaulty agent 
is certainly reachable from $j$
without $i$
between rounds $m'$ and $f+1$; just take $j_{m'} =
\cdots = j_{f+1} = j$.  But even if $j$ fails in round $m'$ according
to $F$, as long $j$ can send a message to a nonfaulty agent other than
$i$, or there
is an appropriate chain of agents, then a nonfaulty agent is
reachable from $j$ without $i$ by round $f+1$.
The probability of there being a failure pattern for which
a nonfaulty agent 
is reachable from $j$ without $i$ depends in part on
how many agents are known to have failed in $I_i$; the more agents are
known not to have failed, the more likely we would expect 
a nonfaulty agent 
to be reachable from $j$ without $i$.

We also want this condition to hold even conditional on a set of
failure patterns, provided that the set of failure patterns does not
favor particular agents failing.  To make this precise, we need a few
more definitions.  
Say that an agent $j$ is \emph{known to be faulty in $I_i$} if
$j$ is faulty in all runs in $\R(I_i)$; thus, $j$ is known to be
faulty in $I_i$ 
if $j$ did not send a message to $i$ at round $m-1$ according to $I_i$.
Say that a set $\F$ of failure patterns \emph{satisfies the
  permutation assumption with respect to a set $F$ of failures and an
  information set $I_i$} if, for all permutations $g$ of the agents that
keep fixed the agents that fail in $F$ or are known to be faulty in
$I_i$, if $F' \in \F$, then so is
$g(F')$, where $g(F')$ is the failure pattern that results by
replacing each triple $(j,m'',A) \in F'$ by $(g(j),m'',g(A))$.
\emph{$\F$ satisfies the permutation assumption with respect to $I_i$} if $\F$
satisfies it with respect to the empty set of failures and $I_i$.  
Let $\R(\F) = \cup_{F \in \F} \R(F)$.

We say that 
$\pi$ \emph{supports reachability} if for all agents $i$, all time-$m$
information sets $I_i$ 
such that $M$ agents are not 
known to be faulty in $I_i$, 
failure pattern $F$,
and all sets $\F$ of failure 
patterns that satisfy the permutation assumption with respect to
$F$ and
$I_i$, we have that  
\begin{enumerate}
\item if $j \ne i$ is not known to be faulty in $I_i$
    and is not in $F$, then 
$$
\begin{array}{l}
\pi(\mbox{no nonfaulty agent 
$l\ne i$ 
    is reachable from $j$ without $i$}
\\ \ \ \ \mbox{between rounds $m$ and $f+1$}  \mid \R(I_i) \cap \R(\F) 
\fullv{\cap \R(F)) \le \frac{1}{2M};}
\shortv{\cap \R(F))\\ \le \frac{1}{2M};}
\end{array} 
$$ 
\item if $j \ne i$ is not known to be faulty in $I_i$
    and is not in $F$, then 
$$
\begin{array}{ll}
    \pi(\mbox{no nonfaulty agent $l \neq i$ is reachable from $j$ without $i$}\\ 
     \ \ \  \mbox{between rounds $m-1$ and $f+1$}
    \mid \R(I_i) \cap \R(\F) 
    \fullv{  \cap \R(F)) \le \frac{1}{2M};}
    \shortv{  \cap \R(F))\\ \le \frac{1}{2M};}
\end{array}
$$  
\item if a message from some agent $j$ not in $F$ was received up to and
  including round $m-2$ but not in round $m-1$,
  then 
$$\begin{array}{ll}\pi(\mbox{no nonfaulty agent 
$l\ne i$}
\fullv{\mbox{ is reachable from an agent $j' \ne i$ that did not receive a message}\\ 
  \ \ \  \mbox{from $j$ in round $m-1$ without $i$ between rounds $m$
    and $f+1$} \mid \R(I_i) \cap \R(\F) 
 \cap \R(F)) \le
  \frac{1}{2M}.}
\shortv{\mbox{ is reachable from an agent}\\ \ \ \ \mbox{$j' \ne i$ that did
    not receive a message  
    from $j$  in round}\\ \ \ \ \mbox{$m-1$ without $i$ between rounds
    $m$ and $f+1$} \\ 
    \ \ \ \ \mid \R(I_i) \cap \R(\F) 
 \cap \R(F)) \le
  \frac{1}{2M}.}
\end{array}$$  
\end{enumerate}
The first two
requirements essentially say that if $i$ hears from $j$ in round
$m-1$, then it is likely that other agents will hear from $j$ as
well in a way that affects the decision, even if $i$ does not forward
$j$'s information. That is, it is unlikely that $j$ will fail right
away, and do so in a way that prevents its information from having an
effect. Similarly, the third requirement says that if $i$ 
does not hear from $j$ in round $m-1$ (as reflected in $I_i$), then
it is likely that other agents will hear
that $j$ crashed at or before round $m-1$ even if $i$ does not report
this fact.


We next define the notion of uniformity.
Given two failure patterns $F^1$ and $F^2$, we say that
$F^1$ and $F^2$ are \emph{equivalent} if there is a permutation $g$ 
of the agents such that $F^2 = g(F^1)$.  
We say that $\pi$ is \emph{uniform}
if, for all equivalent failure patterns $F^1$ and $F^2$
and vectors $\vec{v}$ of initial preferences, 
we have $\pi(F^1,\vec{v}) = \pi(F^2,\vec{v})$.
Intuitively, if $\pi$ is uniform, then the probability of each failure
pattern depends only on the  
number of messages omitted by each agent in each round;
it does not depend on the identity of faulty agents.

\fullv{
The following lemma will prove useful in the argument, and shows where
the uniformity assumption comes into play.
Roughly speaking, the lemma says that if the agents run $\vsigma^{\osc}$, then
 each agent $i$'s expected value
of its initial preference being the consensus value is just its
current knowledge about the fraction of nonfaulty agents that have its
initial preference. The lemma's claim is somewhat stronger,
because it allows for expectations conditional on certain sets
of agents failing.

Before stating the lemma, we need some definitions.  
Let $\R(D_{\ge m})$ consist of all runs 
where a decision is made and the first 
round that seems clean
is $m' \ge m$.
A set $\F$ of failure patterns, a failure pattern $F$, 
a round-$m$ information set $I_i$ for $i$, 
and $m'\ge m$ are
\emph{compatible} if  (a) all the failures in $F$ happen
before round $m'$, 
(b) $m' \le f+1$,
and (c) $\F$ 
satisfies  the permutation assumption with respect to $I_i$ and $F$.
Given an agent $i$ and a run $r$ where consensus is reached,
let $nc(r)$ be the number of agents who apparently have not crashed in
the first round of $r$ that seems clean (i.e., if $m$ is the first
clean round in $r$, then $nc(r) =
|NC_{m}(r)|$), and let $ac(r)$ be the number of these agents in $r$ 
that have initial preference $1$.
Given an information set $I_i \in \I_i$ and  a failure pattern
$F$, 
let $A_F$ be the set of agents who are faulty in $F$; 
let $A$ consist of the agents known to be faulty in $I_i$;
let $n(I_i,F) = n - |A \cup A_F|$; and
let $a(I_i,F)$ be the 
agents not in $A \cup A_F$ that have 
initial preference 1. 
Note that $nc$ and $ac$ are random variables on runs (i.e.,
functions from runs to numbers); 
technically, $a(I_i,F)$ and $n(I_i,F)$ are also random variables on
runs, but $n(I_i,F)$ is constant on runs in $\R(I_i)$, while
$a(I_i,F)$ is constant on runs in $\R(I_i)$ if $m
\ge 2$, since then $I_i$ contains the initial values of nonfaulty agents.


\begin{lemma}
\label{lemma:tree}
If $i$ is an agent who is nonfaulty at the beginning of round 
$m
\le f+1$ and has information set  
$I_i$ (so that $I_i$ is a round-$m$ information set), $F$ is a
failure pattern, $m' \ge m$, $\F$ is a set of failure patterns
such that $\F$, $F$, $I_i$, and $m'$ are compatible, 
$\pi$ is a distribution that supports reachability and is uniform,
and $\pi_{{\vsigma}^{\osc}}(\R(I_i) \cap\R(F) \cap \R(\F)
\cap  \R(D_{\ge m'})) > 0$, then 
\begin{equation}\label{eq:ev}
E[ac/nc \mid \R(I_i) \cap \R(F)\cap \R(\F) \cap
  \R(D_{\ge m'})] = E[a(I_i,F)/n(I_i,F) \mid \R(I_i)], 
\end{equation}
where the expectation is taken
with respect to $\pi_{\vec{\sigma}^{\osc}}$.  
\end{lemma}
\begin{proof}
Let $f' = |A \cup A_F| = n - n(I_i,F)$. For all $f''$ with $f' \le
f'' \le f$, 
let $\R_{f''}$ consists of all runs $r$ where agents are
using $\vec{\sigma}^{\osc}$ such that exactly $f''$
agents are viewed as faulty in the first round that seems clean.  
We claim that, for all $f''$, we have $$E[ac/nc \mid \R_{f''} \cap
\R(I_i) \cap \R(F) \cap \R(\F) \cap \R(D_{\ge m'})] =
E[a(I_i,F)/n(I_i,F) \mid \R(I_i) ).$$ 
Clearly, (\ref{eq:ev}) follows immediately from this claim.

We can calculate the relevant expectations using algebra,  
but there is an easier way to see that the claim holds.
First suppose that $m'> 1$ (so that 
$a(I_i,F)$ and $n(I_i,F)$ are constants on $\R(I_i)$).
If the
first clean round occurs at or after $m'$, then it is easy to see that all
the agents in $A \cup A_F$ will be viewed as faulty in that round (by
all nonfaulty agents), since all these agents fail before round $m'$.
Note that the set of agents viewed as 
faulty in the first clean round of run $r$ is completely determined
by the failure pattern in $r$.  Moreover, it easily follows from the
uniformity assumption, the fact that $\vec{\sigma}^{\osc}$ treats
agents uniformly, and the fact that $\F$ satisfies the permutation
assumption that each set $B$ of cardinality $f''$ that includes $A
\cup A_F$ 
is equally likely to be the set of agents viewed as faulty in the
first clean round of a run in $\R_{f''} \cap \R(I_i) \cap \R(F)
\cap \R(\F) \cap \R(D_{\ge m})$.

Consider the following experiment: choose a set $B$
of $f''$
agents containing
$A \cup A_F$ uniformly at random, and then choose one more
agent $j \notin B$ at random. 
Assign a pair $(B,j)$ value 1 if the agent $j$ chosen has
initial preference 1 in all runs of $I_i$;
otherwise, assign it value 0.  It is easy to see that the
expected value of a 
pair is precisely $E[ac/nc  \mid  \R_{f''} \cap \R(I_i) 
  \cap \R(F)   \cap \R(\F) \cap \R(D_{\ge m})]$.
The $f''$ agents  
in $B$ constitute the set of faulty agents.
The fact that $B$ is  chosen uniformly at
random (among sets of cardinality $f''$ containing $A \cup A_F$)
corresponds to the  
assumption that all choices of $B$ are equally likely.
The last agent chosen determines the
consensus value; as long as there is at least one nonfaulty agent, the
procedure used in runs of $\vec{\sigma}^{\osc}$ 
guarantees that all choices of $j$ are equally likely. 

Now switch the order that the choices are made:
we first choose a nonfaulty agent not in $A \cup A_F$ uniformly at random and
then choose $f''-|A \cup A_F|$ other agents not in $A \cup A_F$ who will
  fail uniformly at 
random. It is clear that there is a one-to-one correspondence between the
choices in the first experiment and the second experiment: in
corresponding choices, the same set
of
$f'' - f'$
agents fail and the same other agent is chosen to determine
the consensus value. 
Moreover, corresponding choices are equally likely. With the second
experiment, it is immediate that the expected value is
$a(I_i,F)/n(I_i,F)$. 

If $m' \le 1$, then the argument is the same, except that the value of
$(B,j)$ is chosen according to the distribution of initial preferences
of agents $j \notin B$ in runs where the faulty agents are exactly the
ones in $B$.
This concludes the proof.
\end{proof}
Theorem~\ref{theo:osc-ne} shows that $\vsigma^{\osc}$ is
a $\pi$-Nash equilibrium, 
as long as $f+1<n$ and 
$\pi$ supports reachability and is uniform.
}

\commentout{
Before proving the result, we need to make clear what agents in a
coalition $K$ know about each other.  Intuitively, we allow arbitrary
pre-play communication.  So we assume that if the deviating agents are
using some strategy $\sigma_K$, then this fact is common knowledge
among them, as are their initial preferences.  As we said, we assume
that these initial preferences are the same (otherwise the coalition
would not have formed).  This means that all agents in a coalition
have the same preferences.  Finally we assume that agents in $K$  can make
their random choices in advance, so we 
can assume that an agent $i\in K$ knows all the random choices of $j
\in K$.  But if $i \in K$ gets some random number $x$ chosen by $l \notin K$,
then $j \in K$ does not know $x$ unless $i$ tells $j$ about it, which
requires a round of communication.
}

\begin{theorem}
\label{theo:osc-ne}
If $f + 1 < n$,
$\pi$ is a distribution that supports reachability, is uniform, and allows up to
$f$ failures, and agents care only about consensus,
then $\vsigma^{\osc}$
is a $\pi$-Nash equilibrium. 
\end{theorem}
\begin{proof}
Fix an agent $i$ and a strategy $\sigma_i$.
We must show that 
we have
\commentout{
\begin{equation}
\label{eq:one-shot}
u_{i^*}(\vsigma^{\osc}) 
\geq u_{i^*}(\vsigma_K,\vsigma^{\osc}_{-i}).
\end{equation}
}
\begin{equation}
\label{eq:one-shot}
u_{i}(\vsigma^{\osc}) \geq u_{i}(\sigma_i,\vsigma^{\osc}_{-i}).
\end{equation}

Suppose, by way of contradiction, that (\ref{eq:one-shot}) does not
hold. Then $i$ must deviate from $\sigma_i^{\osc}$ at some round $m$.
Consider all the ways that $i$ can deviate in round $m$ that
  can affect the outcome (we discuss what it means to affect the
    outcome shortly):


  \begin{enumerate}
\item 
\label{dev-crash}
$i$ pretends to crash; it does not send messages to some
  subset of agents in round $m$ (and then does not not send messages
  from then on).
  \item 
\label{dev-value}
$m=1$ and $i$ sends $(i,1-v_i)$ to some agent $j \neq i$ 
(i.e., $i$ lies about its initial preference to at least one agent).
\item 
\label{dev-message}
$i$ sends an incorrectly formatted message to $j \neq i$
(i.e., $i$ sends a message that is different in format from that
   required by $\vsigma^{\osc}$).

 \item 
\label{dev-pol} 
 $m= 1$
 and $i$ sends values $y_{ij}^t$ to an agent
$j \neq i$ such
   that there is no polynomial $q_i^t$ of degree $1$ that 
   interpolates them all or does not choose the polynomials $q_i^t$ at random.
\item \label{dev:random} $i$ does not choose some $z_{ij}^m$ 
appropriately (as specified by $\vsigma^{\osc}$).

\item 
\label{dev-decide}
$m < f+1$ and $i$ decides on a value in $\{0,1\}$ in round $m$ or
  $m=f+1$ and $i$ 
  decides on an incorrect value
  on the equilibrium path.
  
\item 
\label{dev-pol2}
$m = f+1$ and $i$ sends a value $y_{ji}^t$ to 
$j' \neq i$
  different from the 
  value $y_{ji}^t$ that $i$ received from $j$ in round 1.

\item 
\label{dev-crash-nocrash}
$i$ does not send a round $m'< m$ message to some
agent $j$ 
that $i$ does not know at round $m$ to have been faulty in
round $m'$,
and sends a round $m$ message to 
$j' \ne i$.


\commentout{
\item\label{initpreferencelie} $i$ lies about $j$'s initial preference
  to an agent  
$j' \notin K$; that is, $i$ sends a pair $(j,v_j)$ to $j'$ although
  there is no pair 
  $(j,v_j) \in ST_i$ or it does not send  a pair $(j,v_j)$ to $j'$ although
  there is such a pair in $ST_i$.
}

\item\label{dev-statuslie} $i$ lies about $j$'s status to
$j'\neq i$;
that is, $i$ sends
  $j'$ a status
  report $\overline{SR}_i^m$ such that $\overline{SR}_i^m[j] \ne SR_i^m[j]$.
\end{enumerate}
Note that in a deviation of type 8, we did not consider the case where
$i$ deviates by not sending a message to $j$ in round $m'$ and then
sending a message to $j'$ if $i$ knows that $j$ failed in round $m'$.
In this case, $i$'s deviation is undetectable, and will not affect the
outcome. Clearly if $i$ performs only such undetectable deviations,
then $\sigma_i$ is equivalent to $\sigma_i^{\osc}$, so we do not need
to worry about these deviations.



We consider these deviations one
by one, and show that none of them makes $i$ better off.  More
precisely, 
we show that if $\sigma_i$ involves only deviations
1--$d$ on the list above for appropriate choices of $d$, then
(\ref{eq:one-shot}) holds.
But even this ``brute force'' argument requires some care,
using a somewhat delicate induction on the number of
deviations that $i$ is better off not deviating.
\shortv{We sketch the argument for deviations of type 1 (which turns
  out to be the hardest case) here, leaving the remaining details to
  the full paper.}

\fullv{
We now prove (\ref{eq:one-shot}).
We start with the first type of deviation; that is,
}
\shortv{So}
suppose that
$\sigma_i$ involves only $i$ pretending to crash and
that if $I_i^*$ is a time-$m^*$ information set for $i$,
$\F$ is a set of failure patterns
that satisfies the permutation assumption relative to $I_i^*$,
$\pi_{\vsigma^{\osc}}(\R(I_i^*) \cap \R(\F)) > 0$, and
either there are no deviations in runs in $\R(I_i^*)$ or the first
deviation in a run in 
$\R(I_i^*)$
occurs at or after information set
$I_i^*$,
then
\begin{equation}
\label{eq:better0}
u_{i}(\vsigma^{\osc} \mid \R(I_i^*) \cap \R(\F)) \ge 
u_{i}((\sigma_{i},\vsigma_{-i}^{\osc}) \mid \R(I_i^*) \cap \R(\F)).
\end{equation}
(\ref{eq:one-shot}) clearly follows from (\ref{eq:better0}) by taking
$I_i^*$ to be the initial information set
and letting $\F$ be the set
of all failure patterns compatible with $I_i^*$.

Given a strategy profile $\vec{\sigma}$, let $\R(\vec{\sigma})$ denote
  the possible runs of $\vec{\sigma}$.  
If there are no runs in $\R(\sigma_i,\vsigma^\osc_{-i}) \cap \R(I^*_i)$
in which $i$ pretends to fail, then 
conditional on $\R(I^*_i)$, $\sigma_i$ and $\sigma_i^{osc}$ agree, so
(\ref{eq:better0}) holds.  
If there are runs in $\R(\sigma_i,\vsigma^\osc_{-i}) \cap \R(I_i^*)$
in which $i$ pretends to fail, then 
we proceed by induction on the number of information sets $I_i$ at or after
$I^*_i$ at which $i$ first pretends to crash such that
$\pi_{(\sigma_i,\vsigma^{\osc}_{-i})}(\R(I_i^*) \cap \R(\F) \cap \R(I_i)) > 0$.
Suppose that $i$ first pretends to crash at some information set $I_i$ that
comes at or after $I^*_i$
and 
$\pi_{(\sigma_i,\vsigma^{\osc}_{-i})}(\R(I_i^*) \cap \R(\F) \cap \R(I_i)) > 0$.
Thus, there are no runs in $\R(I_i)$ in which $i$ pretends to crash
prior to information set $I_i$.
Let $\sigma_i'$ be identical to
$\sigma_i$ except that $i$ does not pretend to fail at or after
$I_i$.
By 
(\ref{eq:better0}),
$$u_{i}(\vsigma^{\osc} \mid \R(I_i^*) \cap \R(\F)) \ge
u_{i}((\sigma_i',\vsigma_{-i}^{\osc}) \mid \R(I_i^*) \cap \R(\F)).$$ 
We now show that
\begin{equation}\label{eq:better}
  \begin{array}{ll}
u_{i}((\sigma_i',\vsigma_{-i}^{\osc}) \mid \R(I_i^*) \cap \R(\F))
    \shortv{\\}
\ge 
u_{i}((\sigma_{i},\vsigma_{-i}^{\osc}) \mid \R(I_i^*) \cap \R(\F)).
  \end{array}
  \end{equation}
(\ref{eq:better0}) follows immediately.

To prove (\ref{eq:better}),
since $\R(I_i^*)$ is the union of all the time-$m$ information sets
for $i$ that follow $I_i^*$, 
it suffices to prove that 
for all time-$m$
information sets $I_i'$ for $i$ that follow $I_i^*$, we have 
\begin{equation}\label{eq:better1}
   u_{i}((\sigma_i',\vsigma_{-i}^{\osc}) \mid \R(I_i')
\cap \R(\F)) \ge  
u_{i}((\sigma_i,\vsigma_{-i}^{\osc}) \mid \R(I_i') \cap
\R(\F)) 
\end{equation}
(provided, of course, that 
$\pi_{\vsigma^{\osc}}(\R(I_i') \cap \R(\F)) > 0$;
in the future, we take it for granted 
that the relevant results apply only if we are conditioning on a set
with positive measure).
(\ref{eq:better}) clearly follows from (\ref{eq:better1}),
since the time-$m$ information sets for $i$ partition 
$\R(I_i^*) \cap \R(\F)$.

\commentout{
If $I_i'$ is not permutation-equivalent to $I_i$ 
(i.e., there is no $g$ such that $I_i'$ corresponds to $g(h_i)$),
then 
(\ref{eq:better1}) holds trivially, since
  $\sigma_i'$ agrees with $\sigma_i$ at $I_i'$ and all subsequent
  information sets. 
  If $I_i'$ is permutation-equivalent to $I_i$, then since
  $(\sigma_i',\vsigma^{\osc}_{-i})$ is permutation-invariant and
  $\F$ satisfies the permutation assumption with respect to $I_i'$ 
  (because every agent known to be faulty in $I_i^*$ is also
  known to be faulty in $I_i$, and $\F$ satisfies the permutation
  assumption with respect to $I_i^*$), it follows that
  $i$'s expected utility with $(\sigma_i',\vsigma^{\osc}_{-i})$
  is the same conditional on $\R(I_i)$ as it is conditional on $\R(I_i')$.  
}
If $I_i'\ne I_i$, then   (\ref{eq:better1}) holds trivially,
since in that case   $\sigma_i'$ agrees with $\sigma_i$ at $I_i'$ and
all subsequent information sets. 
  Thus, it suffices to prove
  (\ref{eq:better1}) in the case that $I_i' = I_i$.
We can assume without loss of generality
that $i$'s actions at and after $I_i$ are deterministic. If $i$ is
better off by pretending to fail at $I_i$ 
with some probability, then $i$ is better off by 
pretending to fail at $I_i$ 
with probability 1.  
Note that (a) whether or not there is a seemingly clean round, (b) which is the first
seemingly clean round if there is one, and (c) which agents are considered
nonfaulty at that round are completely determined by the failure
pattern.
Specifically, a particular failure pattern $F' \in \F$
determines the first seemingly clean round $m^*$.
We partition the set 
$\F$
into four sets, 
$\F_1, \ldots, \F_4$, 
and show that
conditional on $\R(I_i) \cap \R(F_j)$,
agent $i$ does at least as well by using $\sigma_i'$ as it does by
using $\sigma_i$, for 
$j = 1,\ldots, 4$.
\shortv{
  \begin{itemize}
    \item $\F_1$ consists of the failure patterns in $\F$
where with $(\sigma_i,\vsigma^\osc_{-i})$ an 
inconsistency is detected (because $f+1$ agents seem to fail).
\item $\F_2$ consists of the failure patterns 
$F' \in \F - \F_1$ such that in all runs 
$r'$
in $\R((\sigma_i',
\vsigma^\osc_{-i})) \cap \R(I_i) \cap \R(F')$, the first
seemingly clean round occurs at some round $m^* < m$. 
\item $\F_3$ consists of the failure patterns in $\F - \F_1$
that result in $m$ being 
  the first seemingly clean round with both
  $(\sigma_i,\vsigma^\osc_{-i})$ and
  $(\sigma_i',\vsigma^\osc_{-i})$.
  \item $\F_4$ consists of the failure patterns in $\F - \F_1$
where the first seemingly clean round $m^*$ with 
$(\sigma_i',\vsigma^\osc_{-i})$ comes at or after $m$ while
with $(\sigma_i,\vsigma^\osc_{-i})$, the first seemingly clean round $m^*$
comes strictly before $m$ or strictly after $m$.
  \end{itemize}
The arguments in the first three cases is relatively straightforward.
To show that, conditional on $\R(I_i) \cap \R(\F_4)$, $i$'s utility is
at least as large with $(\sigma_i',\vsigma_{-i}^{\osc})$ as with
$(\sigma_i,\vsigma_{-i}^{\osc})$ requires the reachability 
and uniformity assumptions.
We leave details to the full paper. \end{proof}}

\fullv{
$\F_1$ deals with a trivial
case; the 
remaining elements of the
partition consider the first seemingly clean round 
of $(\sigma_i',\vsigma^{\osc}_{-i})$ and
$(\sigma_i,\vsigma^{\osc}_{-i})$.
(\ref{eq:better1}) in the case that $I_i' = I_i$ clearly follows from this.

\commentout{
Intuitively, we first consider the failure patterns
where the crash of $i$ in $m$ is not noticed by agents that decide.
Then, we consider the failure patterns where $m^1 < m$,
and the failure patterns where $m^1 \ge m$ and (i) can never happen with $m' < m$.
Finally, we consider the remaining failure patterns, where $m^1 \ge m$
and (i) may happen with $m' < m$. 
}

\begin{enumerate}
  \item[(a)] $\F_1$ consists of the failure patterns in $\F$
where with $(\sigma_i,\vsigma^\osc_{-i})$ an 
inconsistency is detected (because $f+1$ agents seem to fail).
Clearly, conditional on
$\R(I_i) \cap \R(\F_1)$, $i$'s utility is at least as high
with $(\sigma_i',\vsigma_{-i}^{\osc})$ as with 
$(\sigma_i,\vsigma_{-i}^{\osc})$.  It may be that with some failure patterns in
$\F_1$, no inconsistency is detected if $i$ uses
$(\sigma_i',\vsigma^\osc_{-i})$. 
But if the failure pattern is such that an inconsistency is
detected with $\sigma_i'$, then an inconsistency is certainly
detected with $\sigma_i$.  Thus, in all the remaining runs, we
consider no inconsistency is detected with either $\sigma_i$ or $\sigma_i'$.

\item[(b)]
$\F_2$ consists of the failure patterns 
$F' \in \F - \F_1$ such that in all runs 
$r'$
in $\R((\sigma_i',
\vsigma^\osc_{-i})) \cap \R(I_i) \cap \R(F'))$, the first
clean round occurs at some round $m_1 < m$.  It is easy to check that in
a run $r$ of $\R((\sigma_i, \vsigma^\osc_{-i}))$ 
corresponding to
$r'$,
the first clean round also occurs at $m_1$, so that
all agents get the same utility at $r$ and $r'$. Thus, conditional on
$\R(I_i) \cap \R(\F_2)$, $i$'s utility is the same
with $(\sigma_i',\vsigma_{-i}^{\osc})$ and
$(\sigma_i,\vsigma_{-i}^{\osc})$.

\item[(c)] $\F_3$ consists of the failure patterns in $\F - \F_1$
that result in $m$ being 
  the first seemingly clean round with both
  $(\sigma_i,\vsigma^\osc_{-i})$ and
  $(\sigma_i',\vsigma^\osc_{-i})$.  This can happen in runs in
  $\R(\sigma_i,\vsigma_{-i}^{\osc})$ only if the fact
  that $i$ started pretending to fail at $I_i$ with $\sigma_i$ is not
  detected by any agent that does not crash (i.e., if no
  agent that decides is reachable from an agent that does not hear
  from $i$ in round $m$).  Conditional on
$\R(I_i) \cap \R(\F_3)$, $i$'s utility is the same
with $(\sigma_i',\vsigma_{-i}^{\osc})$ and
$(\sigma_i,\vsigma_{-i}^{\osc})$.

\item[(d)] $\F_4$ consists of the failure patterns in $\F - \F_1$
where the first seemingly clean round with 
$(\sigma_i',\vsigma^\osc_{-i})$ comes at or after $m$ while
with $(\sigma_i,\vsigma^\osc_{-i})$, the first clean round $m^*$
comes strictly before $m$ or strictly after $m$. 
Let $M$ be the number of 
  agents that are not known to be faulty in $I_i$, and let $a$ be the
  number of these that share $i$'s initial preference.
It is straightforward to check that $\F_4$ satisfies the permutation
assumption with respect to $I_i$, so 
by Lemma~\ref{lemma:tree},
conditional on $\R(I_i) \cap \R(\F_4)$, 
  $i$'s   expected utility with $(\sigma_i',\vsigma^\osc_{-i})$ is
  $\frac{a\beta_{0i}}{M} + 
  \frac{(M-a)\beta_{1i}}{M}$. 

  To compute $i$'s expected utility with $(\sigma_i,\vsigma_{-i}^{\osc})$,
we must first consider how we could have $m^*$ (the first seemingly
clean round) occur before round $m$.  This can 
happen if (and only if) $i$ first learns
 in round 
 $m''-1 \ge m-1$
  that some agent $j^*$ crashed in round $m' \le
 m-1$, no agent nonfaulty agent $j'$ (other than $i$) will learn
 that $j^*$ crashed in round $m'$ if $i$ pretends to crash,
 and, as a result, round 
$m'$ will seem clean to $j'$. This, in turn can happen if (and only
if) either (i) $m' = m-1$, 
and $i$ does not hear from $j^*$ for
the first time in round $m-1$, 
or (ii) $m' < m$,
$i$ did not hear from $j^*$ for the first time in round
$m'+1$, and there is a 
chain $j_1, \ldots, j_{m''-m'}$ of agents
that ``hides'' the fact that $j^*$ actually crashed
in round $m'$ from $i$ (and all other nonfaulty agents) until round $m''$: 
$j_1$ does not hear from $j^*$ in round $m'$;
for $h < m''-m$, $i$ does not hear from $j_h$ in round $m'+h$;
but $j_{h+1}$ hears from $j_h$ in round $m'+1$
(thus, $j_2$ hears that $j^*$ crashed in round
$m'$ from $j_1$ in round $m'+1$, $j_3$ hears about this from $j_2$
in round $m'+2$, and so on), $i$ hears from $j_{m''-m'}$ in round
$m''$ (and so hears in round $m''$ that $j^*$ crashed in round $m'$);
and there is no shorter chain like this from $j^*$ to $i$.
Note that $i$ can tell by looking at its history at time $m$ whether
it is possible that (i) or (ii) occurrred.  Specifically, (i) can
occur only if there is an agent $j^*$ that $i$ does not hear from for
the first time in round $m$, and (ii) can occur only if there is a
chain $j_1, \ldots, j_{m-m'}$ such that, for $h' < m-m'$, $i$ does not
hear from $j_h$ for the first time in round $m'+h$, either does
not hear from $j_{m-m'}$ in round $m$  or hears from $j_{m-m'}$ that
$j^*$ crashed in round $m'$, and $i$ does not hear that $j^*$
crashed in round $m'$ before round $m$.
Also note that in case (ii), $i$'s history must 
be such that none of the rounds between $m'+1$ and 
$m''-1$ (inclusive)
can seem clean to $i$
(or the other nonfaulty agents).

Agent $i$'s expected
utility with $(\sigma_i,\vsigma^{\osc}_{-i})$ 
conditional on $\R(I_i) \cap \R(\F_4)$ depends on
whether $i$'s history (and hence $I_i$) is such that (i) or (ii) could
have occurred.  If (i) or (ii) could not have occurred, then we must
have $m^* > m$.
  To compute $i$'s expected utility with $(\sigma_i,\vsigma_{-i}^{\osc})$,
  we can apply Lemma~\ref{lemma:tree},
  but now we must include $i$ among the faulty agents (since in the
  first seemingly clean round in runs of
  $\R(\sigma_i,\vsigma_{-i}^{\osc}) \cap \R(I_i) \cap \R(\F_4)$, 
  $i$ will be viewed as faulty by the nonfaulty agents).
  Let $F$ be the failure pattern 
$\{(i,m,A)\}$, where $A$ is the set of agents to which $i$ sends a
message in round $m$ according to $\sigma_i$.  Since
$m^* > m$, we have $\R(I_i) \cap \R(\F_4) \cap \R(F) =
\R(I_i) \cap \R(\F_4) \cap \R(F) \cap \R(D_{\ge m+1})$.  Since
$I_i$, $\F_4$, $F$, and $m+1$ are compatible, by Lemma~\ref{lemma:tree},
$i$'s expected utility with $(\sigma_i, \vsigma_{-i}^{\osc})$
  conditional on
  $\R(I_i) \cap \R(F_4)$ is
  $\frac{(a-1)\beta_{0i}}{M-1} +
  \frac{(M-a)\beta_{1i}}{M-1}$.
  Since $\beta_{1i} > \beta_{2i}$,
conditional on
$\R(I_i) \cap \R(\F_4)$, $i$'s utility is  higher
with $(\sigma_i',\vsigma_{-i}^{\osc})$ than with 
$(\sigma_i,\vsigma_{-i}^{\osc})$.

Now if $I_i$ is such that (i) or (ii) could happen, we use the
reachability assumption to provide upper bounds on the probability
that $m^* < m$.  Note that if (i) holds, $m^* < m$ only if no
nonfaulty agent other than $i$ hears that $j^*$ crashed in round
$m'$.  By part 3 
of the reachability assumption, this happens with probabilty at
most $1/2M$.  If (ii) holds, $m^* < m$ only if there is an appropriate
chain. 
If $m'' = m$,
then agent
$j_{m-m'}$ 
in the chain is not known to be faulty in
$I_i$, so by part 1 of the reachability assumption, the probability
that no nonfaulty agent other than $i$ hears from $j_{m-m'+1}$ that $j^*$
crashed in round $m'$, conditional  on $\R(I_i) \cap \R(\F_4)$ is again
at most $1/2M$.  
Similarly, if $m'' > m$, then $j_{m-m' +1}$ is not known to be faulty in $I_i$,
so by part 2 of the reachability assumption, the probability
that no nonfaulty agent other than $i$ hears from $j_{m-m'+1}$ that $j^*$
crashed in round $m'$, conditional  on $\R(I_i) \cap \R(\F_4)$ is again
at most $1/2M$.
Thus, the probability that $m^* < m$ conditional  on
$\R(I_i) \cap \R(\F_4)$ is at most
$1/M$, even if both (i) and (ii) can occur.
In the runs of $\R(\sigma_i,\vsigma_{-i}^{\osc }) \cap \R(I_i) \cap
\R(\F_4) \cap \R(F)$ where
the first seemingly clean round is $m^* < m$, $i$'s
utility is at most $\beta_{0i}$. If (i) or (ii) could happen and the
first clean round in not before $m$, then it must occur strictly after
$m$, as noted above.  If it does occur after time $m$, then by the
argument above, $i$'s expected utility is   $\frac{(a-1)\beta_{0i}}{M-1} +
  \frac{(M-a)\beta_{1i}}{M-1}$.  Thus, if $I_i$ is such that (i) or
  (ii) could happen, then $i$'s expected utility conditional on
$\R(I_i) \cap \R(\F_4)$
  is
  at most
$$\left(\frac{1}{M} + \frac{M-1}{M} \cdot \frac{a-1}{M-1}\right) \beta_{0i} + \frac{M-1}{M} \cdot \frac{M-a}{M-1} \beta_{1i} = \frac{a}{M} \beta_{0i} + \frac{M-a}{M} \beta_{1i}.$$ 
    In either case,
    conditional on
$\R(I_i) \cap \R(\F_4)$, $i$'s utility is at least as high
with $(\sigma_i',\vsigma_{-i}^{\osc})$ as with 
$(\sigma_i,\vsigma_{-i}^{\osc})$.
  \end{enumerate}

The sets $\F_1, \ldots, \F_4$ form a partition of $\F$: they are
clearly disjoint, and it is not
possible for the first clean round with $(\sigma_i',\vsigma_{-i}^{\osc})$
to be strictly after $m$ while the first clean round with
$(\sigma_i,\vsigma_{-i}^{\osc})$ is $m$.
Thus, we have proved (\ref{eq:better1}) in the case that $I_i' = I_i$,
as desired. This completes the argument for deviations of type 1.
}

\fullv{

Now, consider a deviation of type 2. If $\sigma_i$ is a strategy with
deviations of only types 1 and 2, 
let $\sigma'_i$ be the strategy identical to $\sigma_i$ except that
$i$ does not lie about its initial value
and behaves as if it had not deviated from $\sigma_i$ afterwards.
There is a bijection between runs
of $(\sigma_i, \vsigma_{-i}^{\osc})$ and runs of
$(\sigma_i',\vsigma_{-i}^{\osc})$,
so that two corresponding runs $r$ and $r'$ are identical
except 
that in run $r$ agent $i$ may lie about its initial value and in $r'$ agent $i$
does not. (So, among other things, the random choices made in $r$ and
$r'$ are the same.)  Again, the lie does not affect which round (if
any) will be considered clean nor which agents will be viewed as
nonfaulty in that round. If $i$ is not one of the agents considered
nonfaulty in the clean round, or if $i$ is considered nonfaulty but
$i$ is not the agent whose preference is chosen, then the outcome is
the same in $r$ and $r'$.  If $i$ is the agent whose value is chosen,
then $i$ is worse off if it lies than if it doesn't.   
Thus, $i$ does not gain if it lies about its initial value.  
Again, (\ref{eq:better1}) holds.
Thus, (\ref{eq:better0}) holds for deviations of types 1 and 2.

Finally, we show that (\ref{eq:better0}) holds if we allow deviations of types
\ref{dev-message}--\ref{dev-statuslie}. 
To deal with these,
we proceed by induction on the number of deviations of types
\ref{dev-message}--\ref{dev-statuslie}
in $\sigma_i$,
removing deviations starting
from the earliest deviation.
That is, we consider the information set $I_i$ where the first
deviation of type 3--9 occurs, so that the only deviations prior to
$I_i$ are of type 1 or 2, and show
that we can do better by removing the deviation at $I_i$.
Before getting into the details, we need to state carefully what
counts as a deviation of type 1 or 2 prior to $I_i$.
We try to ``explain'' as much as possible by $i$
pretending to fail, so as to delay the first deviation not of types
1 or 2 as late possible. Thus, if $i$ pretends to fail at
information set $I_i'$
(i.e., sends message according to $\sigma_i^{\osc}$ up to $I_i'$,
sends messages, again according to $\sigma_i^{\osc}$, to some agents
at $I_i'$ and does not send messages to
some agents it does not know to be faulty), and
then sends a message to some agent at some information set $I_i''$
after $I_i'$,
then we say that the first deviation not of types 1 and 2 occurs at $I_i''$
(it is a deviation of type~\ref{dev-crash-nocrash}).

In the base case, $\sigma_i$ contains no deviations of type
\ref{dev-message}--\ref{dev-statuslie};
we have already shown that (\ref{eq:better0}) holds in this case.
For the inductive step, let $I_i$ be an information set at which
$\sigma_i$ has a deviation of type
\ref{dev-message}--\ref{dev-statuslie} and there are no deviations of
type \ref{dev-message}--\ref{dev-statuslie} prior to $I_i$.
We consider each deviation of type
\ref{dev-message}--\ref{dev-statuslie} in turn. 

\begin{enumerate}
  \item[\ref{dev-message}.] If $i$ sends an incorrectly formatted message
  to $j$, then 
  either $j$ receives this message and decides $\bot$ or 
$j$ crashes before sending any messages to an agent $j' \neq i$ (or
before deciding, if $m=f$).
Let $\sigma_i'$ be the strategy
that is identical to $\sigma_i$ except $i$
sends a correctly formatted message to $j$.
In all cases, $i$ does at least as well if 
$i$ 
uses the strategy $\sigma_i'$
as it does 
using $\sigma_i$.
Thus, (\ref{eq:better0}) follows
from the induction hypothesis.

 \item[\ref{dev-pol}.] If $m = 1$ and $i$ sends values $y_{ij}^t$ to 
an agent $j$ such that there is no polynomial $q_i^t$ of
degree $1$ that 
   interpolates them then either an inconsistency is detected or $i$
   would have done at least as well by choosing these values according to
   some polynomial. (Here and in the remainder of the proof, when we
   say ``an inconsistency is detected'', we mean ``an inconsistency is
   detected by a nonfaulty agent different from $i$''.) 
If $i$ does not choose $q_i^t$ at random, since
$f + 1 < n$, there exists a nonfaulty agent $j \neq i$ that sends
   values based on truly random polynomials. 
Thus, the agent whose preference determines the consensus value is
chosen at random, even if $q_i^t$ is not chosen at random.  So
choosing $q_i^t$ at random does not affect
the 
expected
outcome.  
Again, (\ref{eq:better0}) follows from the induction hypothesis.

\item[\ref{dev:random}.] Suppose that $i$ does not choose $z_{ij}^m$ according to protocol.
  From the perspective of an agent $j' \neq i$
  following the protocol $\sigma_{j'}^{\osc}$, it does not affect the
  outcome
  if these values are not chosen randomly.  So, yet again,
  $i$ does just as well if $i$ chooses the numbers randomly, and
  (\ref{eq:better0}) holds.

\item[\ref{dev-decide}.] Clearly there is no benefit to $i$ deciding on a value 
other than $\bot$ 
early (it can
  decide the same value at round $f+1$) and no benefit in deciding an
  incorrect value 
  (since this guarantees that there is no consensus).
  Thus, yet again, (\ref{eq:better0}) holds.  

\item[\ref{dev-pol2}.] Suppose that $m = f+1$ and $i$ lies about $y_{ji}^t$ to some $l \neq i$ for $j \neq i$.
If it turns out that there are not $n-t$ agents that seem to be nonfaulty
in the first clean round, then the value of 
$y_{ji}^t$ is
irrelevant; it is not used in the calculation. If there are $n-t$
seemingly nonfaulty agents in the clean round, then either an
inconsistency is detected due to the lie 
(if $y_{ji}^t$ is sent to
some nonfaulty agent, who then cannot interpolate a
polynomial through it and the other values received), in which case 
$i$ is clearly worse off, or the sum $S$ computed will be a random
element of $\{0, \ldots, n-t-1\}$, so the initial preference of each
of the seeming nonfaulty agents is equally likely to be chosen whether
or not $i$ lies. Thus, $i$ does not gain by 
lying about 
$y_{ji}^t$, so (\ref{eq:better0}) holds.

\item[\ref{dev-crash-nocrash}.] Suppose that $i$ does not send a
  message in round $m' < m$ 
  to an 
  agent $j$  that $i$ does not know (at round $m$) to have been
  faulty at round $m'$ 
  and then $i$ sends a
  message to 
$j' \ne i$ in round $m$.
If $m' < m-1$, since $m$ is the first round that a deviation of types
\ref{dev-message}--\ref{dev-statuslie} occurs, 
and since $i$ does not know at any round $m'' < m$ that $j$ was faulty
at round $m'$ (since $i$ does not know it at round $m$), 
$i$ does not send messages between rounds $m'$ and $m$.
Thus, sending 
a round $m$ message to $j'$ either leads
to an inconsistency being detected 
  or does not affect the outcome
(which can be the case if $j$ fails before 
deciding $\bot$).
This means that $i$ does at least as well if $i$ does not
send a message to $j'$ at round $m$, so
(\ref{eq:better0}) holds.
So we can assume without loss of generality 
that
$m' = m-1$, and that $m'$
is the first round that $i$ did not send a message to an agent $j$.
Similarly, we can assume that $i$ gets a
message from $j'$ in round $m-1$;
otherwise we can consider the strategy $\sigma_i'$ where $i$
does send a message to $j'$ in round $m-1$, and otherwise agrees with
$\sigma_i$, and again the result follows from the induction hypothesis.

The rest of the proof proceeds much in the spirit of the proof for
deviations of type 1.  We partition $\F$ into subsets $\F_1, \ldots,
\F_4$, and show that, for $j = 1, \ldots, 4$, 
$i$ does at least as well with $\vsigma^\osc$ as with $(\sigma_i,
\vsigma_{-i}^\osc)$ 
conditional on $\R(I_i) \cap \R(F_j)$;
(\ref{eq:better0}) then follows.
As in the case of type 1 failures, $\F_1$ consists of the failure
patterns in $\F$ 
where, with $(\sigma_i,\vsigma^\osc_{-i})$, $f+1$ failures are detected.
Clearly, conditional on
$\R(I_i) \cap \R(\F_1)$, $i$'s utility is 
higher
with $(\vsigma^{\osc})$ than with 
$(\sigma_i,\vsigma_{-i}^{\osc})$.  

Let $\F_2$ be the
set of failure patterns 
in $\F - \F_1$
such that
in runs from $\R(I_i) \cap \R(\F_2)$, 
the agents
that decide do not hear about $i$'s round $m$ 
message to $j'$.
Let $\sigma_i'$ be identical
to $\sigma_i$ except that
at $I_i$
agent $i$
does not send a message to $j'$.
It is not hard to
check that
$\F_2$
satisfies the permutation assumption with respect to $I_i$.
Clearly, $i$ gets
the same utility with 
$\sigma_i$ as with $\sigma_i'$ conditional on 
$\R(I_i) \cap \R(\F_2)$.  
Since, with $\sigma_i'$, $i$ has fewer deviations of types 3--9 than
with $\sigma_i$, by the induction hypothesis, 
(\ref{eq:better0}) holds conditional on
$\R(I_i) \cap \R(\F_2)$.

Now let 
$\F_3$
consist of all failure patterns 
in $\F - \F_1$ such that, with $\sigma_i$, the agents that decide
hear both that  $i$ sent a message to $j'$ in round $m$ and 
that $i$ did not send a message to
some agents in round $m-1$.
Thus, with $\sigma_i$, an inconsistency will be detected, 
so $i$ does at least as well with
$\sigma_i'$ as with $\sigma_i$ conditional on 
$\R(I_i) \cap \R(\F_3)$.  
$\F_3$ also 
satisfies the permutation assumption with respect to $I_i$,
so (\ref{eq:better0}) holds conditional on $R(I_i) \cap \R(\F_3)$ by the induction hypothesis.

Finally, let $\F_4$ be the remaining failure patterns in $\F - \F_1$, the ones 
where agents that decide hear about the message sent by $i$ to $j'$
but not about the omissions of $i$ in round $m-1$.
Let $I_i'$ be the round-$(m-1)$ information set preceding $I_i$,
and let $\sigma_i''$ be a strategy identical to $\sigma_i$,
except that at at $I_i'$
$i$ does not deviate from $\sigma_i^{\osc}$.
Conditional on
$\R(I_i) \cap \R(\F_4)$,
$i$ clearly gets the
same utility with $(\sigma_i,\vsigma_{-i}^{\osc})$ as with
$(\sigma_i'',\vsigma_{-i}^{\osc})$.  
It is not hard to show that $\F_4$ also 
satisfies the permutation assumption with respect to $I_i$.
With $\sigma_i''$, $i$ does not deviate at $I_i$, so $i$ has fewer
deviations of types 3--9 than 
with $\sigma_i$.  Thus, by the induction hypothesis, 
(\ref{eq:better0}) holds conditional on $R(I_i) \cap \R(\F_4)$.
This completes the argument for deviations of type 8.

\commentout{
By the reasonableness assumption, if $M$ is the number of agents that
are not known to be faulty in $I_i$, they will hear that $i$ did not
send a message to $j$ in round $m-1$ with probability at least 
$\frac{M-1}{M}$, so $i$'s expected utility with $(\sigma_i,\vsigma_{-i}^{\osc})$
conditional on $\R(I_i) \cap \R(\F_2)$ is
at most $\frac{\beta_{0i}}{M} +   \frac{(M-1)\beta_{2i}}{M}$.  On
the other hand, 
$i$'s expected utility with $(\sigma_i',\vsigma_{-i}^{\osc})$
conditional on 
$\R(I_i) \cap \R(\F_3)$ is
at least 
$\frac{\beta_{0i}}{M} + \frac{(M-1)\beta_{1i}}{M}$.
Thus, (\ref{eq:better2}) holds.
}

\item[\ref{dev-statuslie}.] Suppose that $i$ lies about $j$'s status to 
an agent $j' \neq i$.
That is, either (a) $i$ says that $j$ did
not crash before round $m'$ although $i$ knows that $j$ did crash in
round $m'-1$; (b) $i$ says that $j$ crashed at or before round $m'$ although
$i$ received a message from $j$ in round $m'$ and either $m' = m-1$ or
$m' < m$ and $i$ did not receive a message from any agent saying that 
$j$ crashed in round $m'$; or
(c) $i$ lies about  
the numbers $z_{ji}^{m-1}$ 
sent by $j$ or about which agent reported that
$j$ crashed. Again we consider each of these cases in turn.
We can assume
without loss of generality
that $i$ did not pretend to crash in $I_i$,
since otherwise the arguments for
deviations of
type 8 would apply.
\begin{enumerate}

  \item Suppose that $i$ lies by saying that $j$ did not crash before $m'$ even
    though $i$ knows that $j$ did in fact crash earlier.
    This means that $i$ is claiming to have received a 
    message from $j$ in round $m'$.
    Clearly, it cannot be the case
    that $i$ knows that $j$ crashed before $m'-1$, because then $i$
    would know that no agent would
    get a message 
from $j$
in round $m'-1$, and an inconsistency would be
    detected by $j'$ if the deviation had any impact on the outcome.
    Thus, we can assume that $j$ in fact crashed in round
    $m'-1$. Since we are assuming that $i$ first
deviates in round $m$, 
$i$ must have learned in round $m-1$ about $j$'s crash in round $m'-1$.	    
    That means that either 
    (i) $m'=m$
    and $i$ did not receive a
    message from $j$ in round $m-1$
    or (ii) $m' < m$ and $i$ must 
    have received a message from some agent $j''$ with this
    information in round $m-1$.    
  We can assume without loss of generality that
  $i$ gets a message from $j'$ in round $m-1$, 
  for otherwise $i$ would do at least as well by not lying to $j$, and
    (\ref{eq:better0}) would hold by the induction hypothesis. 

Consider case (i).  If $m=2$, then $i$ pretending that $j$ did not 
crash in round 1 can help only if this leads to round $1$ being viewed as
clean.  But this is the case only if $j'$ received a message from $j$
in round 1 (although $i$ did not). 
According to $\sigma_i^{\osc}$, $i$'s round $m$ message includes the status
report $SR_i^m$. Agent $i$ must send such a status report even with
$\sigma_i$, otherwise an inconsistency is detected and clearly $i$ is
worse off.  Since $i$ claims to have received a
message from $j$ in round 
1,
$SR_i^m[j]$ has the form 
$(\infty,z_{ji}^1)$,
where 
$z_{ji}^{m-1}[i]$ is the random number sent in round $1$ to all agents.
Given that we have assumed that $j$ also sent a round 1 message to
$j'$, $j'$ also received 
$z_{ji}^1[i] = z_{jj'}^1[j']$.
Thus, $j'$ will detect an
inconsistency and decide $\bot$ unless $i$ correctly guesses 
$z_{ji}^1[i]$.
The
probability of $i$ guessing $z_{ji}^1[i]$ correctly is at most $\frac{1}{n}$.

We now partition $\F$ into three sets of failure patterns $\F_1$,
$\F_2$, and $\F_3$, and show that
for $j = 1, 2, 3$, 
$i$ does at least as well with $\vsigma^\osc$ as with $(\sigma_i,
\vsigma_{-i}^\osc)$.  Again, $\F_1$ consists of the failure
patterns in $\F$ where with $(\sigma_i,\vsigma^\osc_{-i})$, $f+1$
failures are detected. Clearly the claim holds in this case.
$\F_2$ consists of the failure patterns $F'$ in $\F - \F_1$ where the 
message that $i$ sent in $I_i$ has no impact on the outcome; that
is, either $i$ crashes before sending the message to $j'$ or no
nonfaulty agent is reachable from $j'$ without $i$ between round $m+1$ and $f+1$.
Let $\sigma_i'$ be identical to $\sigma_i$
except that, at 
$I_i$, 
$i$ replaces the reports relative to $j$ with $SR_i$ (the correct
report) in messages sent to $j'$, while sending the same messages to
other agents. Thus, $i$ has fewer deviations with $\sigma_i'$ than
with $\sigma_i$.  Clearly, conditional on $\R(I_i) \cap \R(\F_2)$, $i$ gets the same expected utility with
$(\sigma_i,\vsigma_{-i}^\osc)$ as with 
$(\sigma_i',\vsigma_{-i}^\osc)$.  It is easy to check that $\F_2$
satisfies the permutation assumption with respect to $I_i$,
so by the induction hypothesis, (\ref{eq:better0}) holds conditional
on $\R(I_i) \cap \R(\F_2)$.  

Let $\F_3$ consist of the remaining failure patterns in 
$\F$.
In runs of $\R(\sigma_i,\vsigma_{-i}^\osc) \cap \R(I_i) \cap \R(\F_3)$, 
$j'$ detects an inconsistency and decides $\bot$ unless $i$ guesses the
random number correctly.
Again, it is not hard to check
that 
$\F_3$ satisfies the
permutation assumption with respect to $I_i$.
Since the largest 
utility that $i$ can get if no inconsistency is detected
is $\beta_{0i}$,  
$$u_{i}((\sigma_i,\vsigma_{-i}^{\osc}) \mid \R(I_i) \cap \R(\F_3))
\leq \frac{1}{n} 
\beta_{0i} + \frac{n-1}{n} \beta_{2i}.$$

On the other hand, by Lemma~\ref{lemma:tree}, 
$$u_{i}(\vsigma^{\osc} \mid \R(I_i) \cap \R(\F_3))
\geq \frac{1}{n} \beta_{0i} + 
\frac{n-1}{n} \beta_{1i}.$$ 
Since $\beta_{1i} > \beta_{2i}$, we have
$$u_{i}(\vsigma^{\osc} \mid \R(I_i) \cap \R(\F_3)) \geq
u_{i}((\sigma_i,\vsigma_{-i}^{\osc}) \mid \R(I_i) \cap \R(\F_3)).
$$
Therefore, (\ref{eq:better0}) holds if $m=2$.

Continuing with case (i), suppose that $m > 2$.  Now it is possible
that $i$ pretending that $j$ did not crash can help even if $j$ did
not send a message to $j'$. Nevertheless, essentially the same
argument will work.  This is because now $SR_i$ would have to include
$z_{ji}^{m-1}$.  Moreover, $z_{ji}^{m-1}[j'] = z_{j'j}^{m-2}[j']$,
the random number in $\{0,\ldots, n-1\}$ sent by $j'$ to $j$ in round
$m-2$. Clearly, $j'$ knows this number, so $i$ would have to guess it
correctly. The argument now proceeds as above.

\commentout{
If $m>2$ and $i$ received $j$'s value in round 1, then (by the
induction assumption), $i$ included $j$'s value in $ST_i^2$, so not
sending this value in later rounds has no impact on the outcomes; that
is, the utility of $i^*$ is the same with $\sigma_i'$ and $\sigma_i$.
If $i$ did not receive $j$'s value in round 1, then $i$ already told
all agents that $j$ was faulty (by the induction assumption), so
including $j$'s value in $ST_i^m$ has no impact on the outcome; $j$'s
value will not be considered in the decision.  
Finally, suppose that $m=2$.  
In that case, the argument is essentially the same as for failures of
type 9.  For all failure patterns $F$ such that no nonfaulty agent is
reachable from $j$ without $i$ by round $f+1$ given $F$, $i$'s message to $j$
makes no difference---$j$'s value will not be included in the decision.
Let $\F$ consist of all failure patterns where an
inconsistency is detected unless $j$ and $j'$ fail (or, if $j \in
K$, $j$ pretends to fail) in such a way that $i$'s message to $j$
makes no difference.  If $i$'s message to $j$ makes no difference,
then the result follows by the induction hypothesis; if $i$ message to
$j$ makes a difference, we can apply the reasonableness assumption
just as we did for failure of type 9. Thus, again 
(\ref{eq:better2}) holds.
}

Now consider case (ii).     
There are two ways in which $i$ can ignore the information that $j''$
sent about $j$ in round $m-1$. The first is to pretend that $j''$ crashed
in round $m-1$; the second is for $i$ to lie about the message that it
received from $j''$ (but to say that it did get a message from $j''$).
In the first case,
as with deviations of type 8,
we can assume without loss of generality that
$i$ does not know that $j''$ is faulty at the beginning round $m$.
We partition $\F$ into three sets much as in the argument for
case (i): 
$\F_1$, the failure patterns in which more than $f+1$  failures are
detected with $\sigma_i$; $\F_2$, the failure patterns where $i$'s lie
has no impact on the outcome; and $\F_3$, the remaining failure patterns.
Again, it is easy to see that (\ref{eq:better0}) holds conditional on
$\R(I_i) \cap \R(\F_1)$ and
$\R(I_i) \cap \R(\F_2)$.  To see that
(\ref{eq:better0}) holds conditional on
$\R(I_i) \cap \R(\F_3)$,
we use the reachability
assumption, much as we did for
as in (d)
of the argument for deviations of type 1.
By part 1 of the reachability assumption, 
if $i$ pretends that $j''$
crashed in round $m-1$,
an inconsistency will be
detected with probability at least $(2M-1)/2M$.  Thus, the same
argument as that used 
in part (e) of the argument for deviations of type 1
shows that
(\ref{eq:better0}) holds conditional on 
$\R(I_i) \cap \R(\F_3)$.

The analysis is essentially
the same if $i$ lies about the message it received from $j''$,
except that, conditional on 
$\R(I_i) \cap \R(\F_3)$,
by the reachability assumption,
$j'$ receives the round $m-1$ message from $j''$
with probability at least $(2M-1)/2M$,
so $j'$ receives inconsistent reports about $j$'s status in round
$m-1$, and decides $\bot$.

\item 
Suppose that $i$ lies to some $j'$ in round $m$ by saying that
$j$ crashed at or before round $m'$ although $i$ received a message from $j$
in round $m'$ and either $m' = m-1$ or $m' < m-1$ and $i$ did not
receive a message from any agent saying that 
$j$ crashed in round $m'$.  If $m' = m-1$, 
then we can proceed as in part (a). Specifically, 
we can use the reachability
assumption to show that $i$ is better off if $i$ does not lie.

The analysis is similar if $i$ pretends to have received
a message in round $m-1$ from some agent $j''$ saying that $j$ crashed
in an earlier round. 
If $i$ did not receive
a message from $j''$ in round $m-1$ 
saying that $j$ crashed before $m'$
but is claiming to have done so, then we can
again use the same arguments as in part (a)
where either $i$ must guess the
random number $z_{j'j''}^{m-2}[j']$ known by $j'$
(if $j''$ did not send a round $m-1$ message to $i$)
or $i$ has to lie about the round $m-1$ report of $j''$.

\item It is easy to see that $i$ does not gain if $i$ lies about
  which agent told him that $j$ crashed or about the 
values $z_{ji}^{m-1}$ sent by $j$ to $i$ in round $m-1$ (and may be
worse off, if an inconsistency is detected).
\end{enumerate}

\end{enumerate}
This completes the proof of the inductive step and, with it, the proof of
the theorem.
\end{proof}
}

\subsection{A $\pi$-Sequential Equilibrium for Fair Consensus}
Our $\pi$-Nash equilibrium requires an agent $i$ to decide on $\bot$
whenever $i$ detects a problem.  While this punishes the agent that
causes the problem, it also punishes $i$.  Would a rational agent
actually 
play such a punishment strategy?  Note that the need to
punish occurs only off the equilibrium path; if all agents follow
$\vsigma^{\osc}$, agents never decide $\bot$.  But to get agents to
play according to $\vsigma^{\osc}$ requires the threat of playing
$\bot$.  There might be a concern that this is an empty threat; a
rational agent might not be willing to play $\bot$ if it detects a deviation.

The solution concept of \emph{sequential equilibrium} \cite{KW82} is a
refinement of Nash equilibrium that, roughly speaking, requires that
agents also make best responses not only on the equilibrium path, but
off the equilibrium path as well.  We now define \emph{$\pi$-sequential
  equilibrium}, a generalization of 
sequential equilibrium
  that allows for faulty agents
(where, as before, $\pi$ is a distribution on failure contexts).  
We then show that $\vsigma^{\osc}$ is essentially a $\pi$-sequential
equilibrium.  

\fullv{\subsubsection{Defining $\pi$-sequential equilibrium}}
\shortv{\paragraph{\rm {\bf Defining $\pi$-sequential equilibrium}}}
Roughly speaking, a strategy profile $\vsigma$ is a sequential
equilibrium if, for each agent $i$ and information set $I_i$ for agent
$i$, 
$\sigma_i$ is a best response to $\vsigma_{-i}$ conditional on
reaching $I_i$ (i.e. conditional on $\R(I_i)$).  The problem is that
the probability of $\R(I_i)$ is 0 if $I_i$ is not on the equilibrium
path, so we cannot condition on $\R(I_i)$. 

Define a \emph{belief system $\mu$} to be a function that associates
with each agent $i$ and information set $I_i$ for agent $i$ a
probability $\mu_{I_i}$ on histories in $I_i$.  
Say that a belief system $\mu$ is \emph{consistent with $\vsigma$
  and $\pi$} if there exists 
a sequence of \emph{completely mixed} strategy profiles
$\vsigma^1,\vsigma^2,\ldots$ 
(where a strategy profile is completely mixed if it gives positive
positive probability to every action at every information set) 
converging to $\vsigma$ such that
$$\mu_{I_i}(h) = \lim_{M \to \infty} \frac{\pi_{\vsigma^M}(h)}{\pi_{\vsigma^M}(I_i)}.$$
Note that $\mu_{I_i}$, $\pi$, and $\vsigma$ together define a probability 
distribution over runs in $\R(I_i)$. 
Let $\mu_{I_i,\pi,\vsigma}$ denote this probability distribution.

A pair $(\vsigma,\mu)$ is a \emph{$\pi$-sequential equilibrium}
if $\mu$ is a belief system consistent with $\vsigma$ and $\pi$
such that, for every agent $i$,
information set $I_i$, and strategy $\sigma_i'$,
$u_i((\sigma_i,\vsigma_{-i}) \mid \R(I_i)) \ge u_i((\sigma_i',\vsigma_{-i}) \mid \R(I_i))$,
where now the expected utility is taken with respect to
$\mu_{I_i,\pi,\vsigma}$. (Kreps and Wilson's \citeyear{KW82} definition
of sequential equilibrium is 
identical, except that there is no distribution $\pi$ on failure contexts.)

\fullv{\subsubsection{Extending $\vsigma^{\osc}$ to a $\pi$-sequential
    equilibrium}}
\shortv{\eject}
\shortv{\paragraph{\rm {\bf Extending $\vsigma^{\osc}$ to a $\pi$-sequential equilibrium}}}
We now show that the protocol $\vsigma^{\osc}$ can be extended to 
a $\pi$-sequential equilibrium with minimal changes.
In the proof of Theorem~\ref{theo:osc-ne}, we showed that $i$ could
not gain by deviating at an information set $I_i$ where there were 
no deviations 
of type 1--9
prior to $I_i$.
We did not show that $i$ does not gain from deviating at $I_i$
if an inconsistency is detected at $I_i$, so that $i$ is expected to decide $\bot$.
In fact, if $i$ believes that the inconsistency may go unnoticed by
other agents due to crashes and consensus may still be reached on some value in $\{0,1\}$,
then $i$ always gains by not deciding $\bot$.
However, suppose that $\mu^{\se}$ is a  belief system
such that 
at an information set $I_i$ for $i$ that is off the equilibrium path due
to a deviation (or multiple deviations) from 
$\vsigma^{\osc}$ by agents other
than $i$, $i$ believes that these agents decided $\bot$ when
they deviated.  (Intuitively, $i$ believes that if the agents were
crazy enough to deviate in the first place, then they were also crazy
enough to decide $\bot$.)   In that case, deciding $\bot$ is also a
best response for $i$.

%

The belief system $\mu^{\se}$ is not enough to deal with information
sets $I_i$ off the equilibrium path due to $i$ himself having
deviated.  Agent $i$ cannot believe that it played $\bot$ when
it in fact did not.  To get a sequential equilibrium,
we modify $\sigma_i^{\osc}$ at information sets off the equilibrium
path
that are reached due only to agent $i$'s deviations.
Define the strategy  $\sigma^{\se}_i$ so that it agrees with $\sigma_i^{\osc}$
at every information set $I_i$ where agent $i$ has not
deviated in the past.
Thus, in particular, $i$ decides $\bot$ with 
$\sigma_i^{\se}$ if $i$
detects an inconsistency at one 
of these information sets.  More generally, say that an information
set $I_i$ is
\emph{unsalvageable} if $i$ knows at $I_i$
that another agent $j$ 
deviated or detected an inconsistency at a point when
$j$ had not crashed,
and thus decided $\bot$.
$I_i$ is certainly unsalvageable if reaching $I_i$
requires deviations by agents other than $i$ (for then the agent
that performed that deviation 
decided $\bot$).
But even if $i$ is the only agent who deviates at $I_i$, $I_i$ may be
unsalvageable.  For example,  $i$ does not send a message to
$j$ in round $m_1$, $i$ sends a message to $j$ in round $m_2>m_1$, and then
$j$ sent a message to $i$ in round $m_2+1$, 
the round-$(m_2+2)$
information set where $i$ receives $j$'s message is also
unsalvageable.  If $I_i$ is unsalvageable, $i$ decides $\bot$.
Finally, if $I_i$ is salvageable, then at $I_i$ agent $i$ acts in
a way that is most likely to have the other agents think that there
has been no inconsistency.  In general, there may be more than one
failure pattern that will prevent a nonfaulty agent from realizing
that there is an inconsistency.  For example, if $f=1$, $n=3$, and
agent $1$ did not send a message to agent $2$ in round $m$, but did
send a message to agent 3, then $i$ can either not send a message to
any agent in round $m+1$, or it can send a message to agent 3.
If it is more likely that neither 2 nor 3
failed in round $m$ than agent 2 failed before telling agent 3
that it did not hear from 1, then it would be better for $i$ 
not to send a message to 2 or 3 in round $m+1$.  
If there is more than one best response, then $i$ chooses
a fixed one according to some ordering on actions.
(Note that this means that, unlike $\vsigma^{\osc}$, the behavior of
$\vsigma^{\se}$ may depend on $\pi$.)

Having defined $\vsigma^{\se}$, we can now define $\mu^{\se}$
formally.  We assume that there are only 
finitely many actions that $i$ can play at each 
of its information sets $I_i$: it can send one of $K_{I_i}$ possible
messages and/or decide one of 
$\bot$, 0, or 1 if it has not yet made a decision, or do nothing.
Given an integer $M>0$, let $\vsigma^M$ 
be the strategy profile where at each information set $I_i$,
agent $i$ 
plays 
$\sigma_i^{\se}(I_i)$ 
with probability $1-1/M$, and
divides the remaining probability $1/M$ over all the actions
that can be played at $I_i$ as follows: if $i$ has already decided 
before, then $i$ sends each of the 
$K_{I_i}$
possible messages with equal
probability 
$\frac{1}{M(K_{I_i}+1)}$ 
and does nothing with probability
$\frac{1}{M(K_{I_i}+1)}$;
if $i$ has not yet decided at $I_i$, then for 
each of the 
$K_{I_i}$
messages {\bf m} that it can send, it decides $\bot$
and sends {\bf m} with probability 
$\frac{1}{M(K_{I_i}+1)} - \frac{1}{M^2(K_{I_i}+1)}$,
decides $\bot$ and sends no message with
probability 
$\frac{1}{M(K_{I_i}+1)} - \frac{1}{M^2(K_{I_i}+1)}$, 
and performs
each of the remaining
$3(K_{I_i}+1)$ 
possible actions with equal
probability
$\frac{1}{3M^2(K_{I_i}+1)}$.  
Clearly $\vsigma^M$ is completely mixed
and the sequence $\vsigma^M$ converges to $\vsigma^{\se}$.  
Given a round-$m$ information set $I_i$ and global history $h \in I_i$, let
$$\mu_{I_i}^{\se}(h) = \lim_{M \to \infty}
\frac{\pi_{\vsigma^{M}}(h)}{\pi_{\vsigma^{M}}(I_i)}.$$ 
The effect of this definition of $\mu_{I_i}^{se}$  beliefs is that if
$I_i$ is off the equilibrium path as a result of
some other agent $j$'s deviation, then $i$ believes that $j$ played
$\bot$.  Moreover, $i$ believes that other agents $j$ have similar beliefs.

\fullv{
Theorem~\ref{theo:osc-se} shows that $\vsigma^{\se}$ is a 
$\pi$-sequential equilibrium for a reasonable and uniform $\pi$. 
}

\begin{theorem}
\label{theo:osc-se}
If $f + 1 < n$, $\pi$ is a distribution that supports reachability, is
uniform, and allows
up to $f$ failures,  
and agents care only about consensus, then $(\vsigma^{\se},\mu^{\se})$ is a
$\pi$-sequential equilibrium. 
\end{theorem}
\fullv{
\begin{proof}
Fix an agent $i$, a round-$m$ information set $I_i$, and strategy $\sigma_i$.
It is easy to see that $\mu^{\se}$ is consistent. Thus, it suffices to show that
\begin{equation}
\label{eq:osc-se}
u_{i}((\sigma_i^{\se},\vsigma^{\se}_{-i}) \mid \R(I_{i})) \ge u_{i}((\sigma_i,\vsigma_{-i}^{\se}) \mid \R(I_{i})).
\end{equation}
%
We need to consider the cases where (a) $I_i$ is consistent with $\vsigma^{\se}$,
(b) $I_i$ is inconsistent with $\vsigma^{\se}$ and unsalvageable, 
and (c) $I_i$ is inconsistent with
$\vsigma^{\se}$ and salvageable. 
In case (a), $\sigma^{\se}_i$ agrees with $\sigma_i^{\osc}$; the
argument of the proof of Theorem~\ref{theo:osc-ne} shows that it is a
best response.  In case (b), the definition of $\mu^{\se}$ guarantees
that $i$ ascribes probability 1 to whichever agent has deviated or
detected a deviation playing $\bot$, so it is a best response for $i$
to play $\bot$.   Finally, in case (c), for failure patterns where
some other agent $j$ detects $i$'s deviation, $i$ ascribes probability 1
to $j$ playing $\bot$, so it does not matter what $i$ does.  On the
other hand, for failure patterns where all the nonfaulty agents will
consider it possible that there are no deviations, the proof of
Theorem~\ref{theo:osc-ne} shows that $i$ should continue to play in a
way consistent with $\sigma_i^{\osc}$.  If there are several choices
of how to play that might be consistent with $\sigma_i^{\osc}$, then
$i$ should clearly play one that is best.
\end{proof}
}

\section{Discussion}
\label{sec:discussion}
We have provided a strategy for consensus that is a $\pi$-Nash equilibrium
and can be extended to a $\pi$-sequential equilibrium, where $\pi$ is
a distribution on contexts that allows up to $f$ failures and
satisfies minimal conditions, as long as $n > f+1$.  
Although our argument is surprisingly complicated, we have considered
only the simplest possible case: synchronous systems, crash failures,
and only one player deviating (i.e., no coalitions).
A small variant of
our strategy also gives a Nash and sequential equilibrium even if
coalitions are allowed, but proving this seems significantly more
complicated.
We are currently writing up the details carefully.
Of course, things will get even worse once we allow more
general types of failures, such as omission failures and Byzantine
failures.  But such failure types, combined with rational agents, are
certainly of interest if we want to apply consensus in, for example,
financial settings of the type considered by Mazi{\`e}res
\citeyear{Maz15}.  Consensus is known to be impossible in an
asynchronous setting, even with just one failure \cite{FLP}, but
algorithms that attain consensus with high probability are well known
(e.g., \cite{Aspnes03}).  We may thus hope to get an
$\epsilon$--$\pi$-Nash equilibrium in the asynchronous setting if we
also allow rational agents.  We believe that the techniques developed in
this paper will be applicable to these more difficult problems.

It is also worth examining our assumptions regarding distributions in
more detail.  The uniformity assumption implies that no agent is more
likely to fail than any other.  If all agents can be identified with
identical computers, then this seems quite reasonable.  But if one
agent can be identified with a computer that is known to be more prone
to failure, then the uniformity assumption no long holds.    Note
that the uniformity assumption does allow for correlated failures,
just as long as the permutation  of a correlated failure is just as
likely as the unpermuted version.

Now consider the assumption that $\pi$ supports reachability.   If we are
considering Nash equilibrium (where there is only one deviating
agent), the 
assumption says that the probability, conditional on an information
set $I_i$ (and some assumptions about failures), that
some information (about a 
message sent by an agent that crashes or about the fact that an agent
crashed in a particular round) is quite high, where ``quite high'' is
a function of the number of agents $M$ that are nonfaulty according to
$I_i$.  Since the more nonfaulty agents there are, the more likely it
is that an agent $l \ne i$ is reachable from $j$ without $i$.  

\fullv{
However, once we allow coalitions $K$ of agents, it becomes less
likely that $l \notin K$ is reachable from $j$ without $K$ with
 probability $1/2M$, not taking $K$ into account. To take an
extreme example, suppose that $|K| = k$, $f=1$, and $n =
k+2$.  Now suppose that $i$ receives  a message from agent $j''$ in
round 3 that some other agent $j$, from whom $i$ got a round 1
message, crashed in round 1.  Further suppose that round 1 would be
considered clean if $j$ did not crash in round 1 and that $i$'s utility
would be higher if round 1 is considered clean rather than a later
round.  Thus, it may be to $i$'s benefit not to forward $j$'s
message; if $j$ in fact crashes without any nonfaulty agent hearing 
$j$'s message, $i$ will be better off.  Since all the agents in $K$
can coordinate in not forwarding $j$'s message, $j$'s message will
reach a nonfaulty agent only if either $j$ is nonfaulty, or $j$
crashes either after round 2 or crashes at round 2, but still sends a
message to the nonfaulty agent that is not in $K$ before crashing.
Since $M=n$ in this case, this means that $j$'s message must reach a
nonfaulty agent with probability at least $\frac{2n-1}{2n} =
\frac{4k+3}{4k+4}$, independent of $k$.  For small $k$, this seems
quite reasonable; for large $k$, it does not.  This suggests that this
assumption is appropriate 
if $k+f$ is not 
a large
fraction of $n$.
}

Our final comment concerns the fairness assumption.  While this
assumption distinguishes our work from some of the other related work
(e.g., \cite{AGLS14,BCZ12}), since, as we observed above, a consensus
protocol must essentially implement a randomized dictatorship,
achieving fairness  once we get consensus in the
presence of rational and faulty agents is not that difficult; we must
simply ensure that the 
rational agents cannot affect the probability of a particular agent
being selected as dictator.  We enforce this using appropriate
randomization in our protocol.  The requirement in \cite{BCZ12} that
consensus must be achieved no matter what the deviating agents do
turns out to have far more impact on the technical results than the
fairness requirement.  

In any case, we believe that the need for dealing with both rational
and faulty agents in consensus protocols is compelling. There is
clearly much more to be done on this problem.  
%

\shortv{\bigskip \noindent {\bf Acknowledgments:} Joe Halpern's work
  supported in 
part by NSF grants  IIS-0911036 and CCF-1214844, and by AFOSR grant
FA9550-12-1-0040,
and ARO grant W911NF-14-1-0017. 
Xavier Vila\c{c}a has been partially supported by Funda\c c\~{a}o para
a Ci\^{e}ncia e Tecnologia (FCT)  
through projects with references PTDC/EEI-SCR/1741/2014 (Abyss), UID/CEC/50021/\-20\-13,
and ERC-2012-StG-307732, and through the PhD grant SF\-RH\-/BD\-/79822/2011.
\vspace{.1in}}


\fullv{\bibliographystyle{chicago}}
\shortv{\bibliographystyle{abbrv}}
\bibliography{z,joe}

\begin{thebibliography}{10}
\shortv{\vspace{.05in}}
\bibitem{ADGH06}
I.~Abraham, D.~Dolev, R.~Gonen, and J.~Y. Halpern.
\newblock Distributed computing meets game theory: robust mechanisms for
  rational secret sharing and multiparty computation.
\newblock In {\em Proc.~25th ACM Symposium on Principles of Distributed
  Computing}, pages 53--62, 2006.

\bibitem{ADH13}
I.~Abraham, D.~Dolev, and J.~Y. Halpern.
\newblock Distributed protocols for leader election: a game-theoretic
  perspective.
\newblock In {\em Proc.~27th International Symposium on Distributed Computing},
  pages 61--75, 2013.

\bibitem{AGLS14}
Y.~Afek, Y.~Ginzberg, S.~Landau, and M.~Sulamy.
\newblock Distributed computing building blocks for rational agents.
\newblock In {\em Proc.~33rd ACM Symposium on Principles of Distributed
  Computing}, pages 406--415, 2014.

\bibitem{Alvisi05}
A.~S. Aiyer, L.~Alvisi, A.~Clement, M.~Dahlin, J.~P. Martin, and C.~Porth.
\newblock {BAR} fault tolerance for cooperative services.
\newblock In {\em Proc.~20th ACM Symposium on Operating Systems Principles
  (SOSP 2005)}, pages 45--58, 2005.

\bibitem{Aspnes03}
J.~Aspnes.
\newblock Randomized protocols for distributed consensus.
\newblock {\em Distributed Computing}, 16(2--3):165--176, 2003.

\bibitem{BCZ12}
X.~Bei, W.~Chen, and J.~Zhang.
\newblock Distributed consensus resilient to both crash failures and strategic
  manipulations.
\newblock Available at http://arxiv.org/abs/1203.4324; version 3, 2012.

\bibitem{Bp03}
E.~Ben-Porath.
\newblock Cheap talk in games with incomplete information.
\newblock {\em Journal of Economic Theory}, 108(1):45--71, 2003.

\bibitem{DS1}
D.~Dolev and H.~R. Strong.
\newblock Polynomial algorithms for multiple processor agreement.
\newblock In {\em Proc.~14th ACM Symposium on Theory of Computing}, pages
  401--407, 1982.

\bibitem{FLP}
M.~J. Fischer, N.~A. Lynch, and M.~S. Paterson.
\newblock Impossibility of distributed consensus with one faulty processor.
\newblock {\em Journal of the ACM}, 32(2):374--382, 1985.

\bibitem{Gibbard73}
A.~Gibbard.
\newblock Manipulation of voting schemes.
\newblock {\em Econometrica}, 41:587--602, 1973.

\bibitem{Gibbard77}
A.~Gibbard.
\newblock Manipulation of schemes that mix voting with chance.
\newblock {\em Econometrica}, 45(3):665--681, 1977.

\bibitem{GKTZ12}
A.~Groce, J.~Katz, A.~Thiruvengadam, and V.~Zikas.
\newblock Byzantine agreement with a rational adversary.
\newblock In {\em Proc.~39th Internation Colloquium on Automata, Languages, and
  Programming, Part {II}}, pages 561--572, 2012.

\bibitem{HT04}
J.~Y. Halpern and V.~Teague.
\newblock Rational secret sharing and multiparty computation: extended
  abstract.
\newblock In {\em Proc.~36th ACM Symposium on Theory of Computing}, pages
  623--632, 2004.

\bibitem{HV16}
J.~Y. Halpern and X.~Vilaca.
\newblock Rational consensus.
\newblock available at www.cs.cornell.edu/home/halpern/papers/ratconsensus.pdf,
  2016.

\bibitem{KW82}
D.~M. Kreps and R.~B. Wilson.
\newblock Sequential equilibria.
\newblock {\em Econometrica}, 50:863--894, 1982.

\bibitem{Lyn97}
N.~A. Lynch.
\newblock {\em Distributed Algorithms}.
\newblock Morgan Kaufmann, San Francisco, 1997.

\bibitem{Maz15}
D.~Mazi{\`e}res.
\newblock The {S}tellar consensus protocol: a federated model for
  internet-level consensus.
\newblock Available at www.stellar.org/papers/stellar-consensus-protocol.pdf,
  2015.

\bibitem{Satter75}
M.~Satterthwaite.
\newblock Strategy-proofness and {A}rrow's conditions: existence and
  correspondence theorems for voting procedures and social welfare functions.
\newblock {\em Journal of Economic Theory}, 10:187--217, 1975.

\bibitem{shamir}
A.~Shamir.
\newblock How to share a secret.
\newblock {\em Communications of the ACM}, 22:612--613, 1979.

\end{thebibliography}

\end{document}